\documentclass[letterpaper, 10pt, journal, final]{IEEEtran}  
\IEEEoverridecommandlockouts 
\usepackage{mathtools} 
\usepackage{graphicx} 
\usepackage{amsmath} 
\usepackage{amssymb}  
\usepackage{cite}
\usepackage{multirow}
\usepackage{color}
\usepackage{hyperref}
\usepackage{algorithm, algorithmicx, algpseudocode}

\newcommand{\B}{{\mathbb B}}

\newcommand{\A}{{\mathbb A}}
\newcommand{\R}{{\mathbb R}}

\newcommand{\N}{{\mathbb N}}
\newcommand{\Rnn}{{\mathbb R}_{\ge 0}}
\newcommand{\Rp}{{\mathbb R}_{> 0}}
\newcommand{\C}{{\mathbb C}}

\newcommand{\cA}{{\mathcal A}}

\newcommand{\cB}{{\mathcal B}}

\newcommand{\cD}{{\mathcal D}}

\newcommand{\cF}{{\mathcal F}}
\newcommand{\cG}{{\mathcal G}}

\newcommand{\cI}{{\mathcal I}}

\newcommand{\cK}{{\mathcal K}}

\newcommand{\cM}{{\mathcal M}}

\newcommand{\cO}{{\mathcal O}}
\newcommand{\cP}{{\mathcal P}}
\newcommand{\cQ}{{\mathcal Q}}

\newcommand{\cT}{{\mathcal T}}

\newcommand{\diag}{{\mathrm{diag}}}

\newcommand{\cl}{\mathrm{cl}}

\newcommand{\bfone}{\mathbf 1}

\newcommand{\inter}{\mathrm{int}}
\newcommand{\dom}{\mathrm{dom}}

\def\QED{\mbox{\rule[0pt]{1.3ex}{1.3ex}}} 
\newenvironment{proof}{{\quad \it Proof:\,}}{\hfill \QED \par}
\newenvironment{proof-of}[1]{{\quad\it Proof of #1:\,}}{\hfill\QED\par}

\newtheorem{thm}{Theorem}
\newtheorem{cor}{Corollary}

\newtheorem{defn}{Definition}

\newtheorem{prop}{Proposition}

\newtheorem{rem}{Remark}

\title{Geometric Properties of Isostables and Basins of Attraction of Monotone Systems
\author{Aivar Sootla and Alexandre Mauroy
\thanks{Aivar Sootla is with the Department of Engineering Science, University of Oxford, Parks Road, Oxford, OX1 3PJ, UK {\tt aivar.sootla@eng.ox.ac.uk}.}
\thanks{Alexandre Mauroy is with Namur Center for Complex Systems (naXys) and Department of Mathematics, University of Namur, B-5000, Belgium {\tt alexandre.mauroy@unamur.be}}
\thanks{The authors would like to thank Prof Zhao for explaining some aspects of the results in~\cite{jiang2004saddle}. Most of this work was performed, while Aivar Sootla and Alexandre Mauroy were with Montefiore Institute, Li\`ege University and were supported by F.R.S.-- FNRS postdoctoral grant and a BELSPO Return Grant, respectively. Currently, Aivar Sootla is supported by the EPSRC Grant EP/M002454/1.
}
}
}
\begin{document}
\maketitle
\IEEEpeerreviewmaketitle
\begin{abstract}
In this paper, we study geometric properties of basins of attraction of monotone systems. Our results are based on a combination of monotone systems theory and spectral operator theory. We exploit the framework of the Koopman operator, which provides a linear infinite-dimensional description of nonlinear dynamical systems and spectral operator-theoretic notions such as eigenvalues and eigenfunctions. The sublevel sets of the dominant eigenfunction form a family of nested forward-invariant sets and the basin of attraction is the largest of these sets. The boundaries of these sets, called isostables, allow studying temporal properties of the system. Our first observation is that the dominant eigenfunction is increasing in every variable in the case of monotone systems. This is a strong geometric property which simplifies the computation of isostables. We also show how variations in basins of attraction can be bounded under parametric uncertainty in the vector field of monotone systems. Finally, we study the properties of the parameter set for which a monotone system is multistable. Our results are illustrated on several systems of two to four dimensions. 
\end{abstract}
\section{Introduction\label{s:intro}}
In many applications, such as economics and biology, the states of linear dynamical systems take only nonnegative values. These systems are called \emph{positive} and have received considerable attention in the context of systems theory~\cite{coxson1987positive,PosSysBook}, model reduction~\cite{Sootla2012positive, grussler2012symmetry}, distributed control~\cite{rantzer2015ejc,tanaka2011bounded}, etc. One of the main tools to study such systems is the Perron-Frobenius theorem (see e.g.~\cite{berman1994nonnegative}), which describes some spectral properties of the drift matrix in a linear positive system. In the nonlinear setting, positive systems have a couple of generalizations, one of which is the class of \emph{cooperative monotone systems} (see e.g.~\cite{hirsch2005monotone}). Similarly to the linear case, cooperative monotone systems generate trajectories (or flows) which are increasing functions in every argument with respect to the initial state and for every time. With a slight abuse of terminology, we will refer to cooperative monotone systems simply as \emph{monotone}. Over the years there were a number of developments in monotone systems theory \cite{dirr2015separable,gouze2000interval} as well as applications such as finance~\cite{bergenthum2007comparison}, energy networks~\cite{zlotnik2015flows}, ventilation systems~\cite{meyer2013controllability}, biology~\cite{angeli2004detection,angeli2010graph,rami2014stability}, etc.

The Koopman operator (see e.g.~\cite{budivsic2012applied}) provides a framework that allows to define and study spectral properties of dynamical systems in the basins of attraction of their hyperbolic attractors. The operator has eigenvalues and eigenfunctions (i.e., infinite dimensional eigenvectors), which are directly related to the geometric properties of the system. For instance, the level sets of the dominant eigenfunction (that is, the eigenfunction corresponding to the eigenvalue with the maximal real part) are called \emph{isostables}~\cite{mauroy2013isostables} and contain the initial conditions of trajectories that converge synchronously toward the fixed point. In addition, the interior of the sublevel set at infinity is the basin of attraction of the fixed point. Hence isostables serve as a convenient refinement of the basins of attraction and add further details to the geometric description of the system. 

In~\cite{hirsch2005monotone}, it was mentioned that the flow of a monotone system can be seen as a positive operator. Hence the authors argued that an operator version of the Perron-Frobenius theorem, which is called the Krein-Rutman theorem, can be applied. However, the investigation into spectral properties of these operators lacked, probably since spectral theory of such operators was not well-developed. This gap can be filled by the Koopman operator framework, which may pave the way to formulate a version of the Perron-Frobenius theorem for monotone systems admitting a stable hyperbolic fixed point.

In this paper, we first provide a spectral characterization of a so-called maximal Lyapunov function, which is used to compute basins of attraction in~\cite{vannelli1985maximal}. To do so we use the eigenfunctions of the Koopman operator. In the case of monotone systems, we show that the maximal Lyapunov functions can be constructed with only one eigenfunction under some mild assumptions. We proceed by studying the properties of the isostables, which we connect to the properties of basins of attraction. Basins of attraction have been extensively studied in the case of monotone systems~\cite{takac1991domains,smith2001stable,jiang2004saddle}. In~\cite{sootla2016basins}, we showed that the isostables of monotone systems have properties similar to the boundaries of basins of attraction. In this paper, we expand these arguments using properties of general increasing functions and order-convexity. Order-convexity is a strong geometric property that is well-suited to describe the behavior of monotone systems.

We proceed by studying systems with two asymptotically stable fixed points (i.e. bistable systems). We consider a class of bistable (not necessarily monotone) systems, whose vector fields can be bounded from below and above by the vector fields of two bistable monotone systems. For this class of systems, basins of attraction can be estimated using the basins of these monotone bounding systems. We note that the idea of bounding a system with monotone ones is not novel and appears in many works (see e.g.,~\cite{gouze2000interval,ramdani2010computing}). This approach is then extended to estimate basins of attraction of monotone systems under parametric uncertainty. A preliminary study on estimating basins of attraction under parameter uncertainty was performed in~\cite{sootla2016basins}, and we generalize it in this paper by providing easy to verify assumptions. Furthermore, we study the properties of the parameter set for which a monotone system is (at least) bistable. We illustrate our theoretical findings with several numerical examples. 

Our theoretical results are complemented by a discussion on methods for computing inner and outer estimates on basins of attraction. We cover methods based on linear algebra~\cite{mauroy2014global} and sum-of-squares programming~\cite{henrion2014convex, valmorbida2014region}, and their relation to the Koopman operator framework. However, in the case of monotone systems, we propose to exploit a different technique. According to properties of monotone systems, we can build inner and outer approximations of the basin of attraction by computing flows starting from a finite number of points. This allows us to derive data-sampled algorithms. We discuss two conceptually similar algorithms exploiting this idea~\cite{sootla2016basins, kim2016directed}.

We also exploit some results and techniques from~\cite{sootla2015pulsesaut}, which considers the problem of switching between the stable fixed points of a bistable monotone system by using a pulse control signal. In that work, the authors introduced the concept of switching set, which is reminiscent of basins of attraction, but in the space of control parameters. In~\cite{sootla2016nolcos}, the concept of eigenfunctions was extended to the pulse control problem. 

The rest of the paper is organized as follows. In Section~\ref{s:prel}, we introduce the main properties of the Koopman operator and monotone systems. In Section~\ref{s:order-convex} we present the properties of order-convex sets, which provide basic but strong topological tools for monotone systems. In Subsection~\ref{s:spec-mon} we discuss spectral properties of monotone systems and build (maximal) Lyapunov functions using eigenfunctions of the Koopman operator. We study geometric properties and investigate the behavior of basins of attraction of monotone systems under parameter variations in Subsections~\ref{ss:geometry} and~\ref{ss:basins}, respectively. We discuss methods to compute isostables and basins of attraction in Section~\ref{s:alg-disc}. We provide numerical examples in Section~\ref{s:examples} and conclude in Section~\ref{s:conc}. 

\section{Preliminaries\label{s:prel}}
Throughout the paper we consider parameter-dependent systems of the form
\begin{equation}
\label{sys:f}
\dot x = f(x,p),\quad x(0) = x_0,
\end{equation} 
with $f: \cD\times \cP\rightarrow \R^n$,  $\cD\subset\R^n$, and $\cP\subset\R^m$ for some integers $n$ and $m$. We define the flow map $\phi_f: \R \times \cD \times \cP\rightarrow \R^n$, where $\phi_f(t, x_0, p)$ is a solution to the system~\eqref{sys:f} with the initial condition $x_0$ and the parameter $p$. We assume that $f(x,p)$ is continuous in $(x,p)$ on $\cD\times \cP$ and twice continuously differentiable in $x$ for every fixed $p$, unless it is stated otherwise. We denote the Jacobian matrix of $f(x, p)$ as $J(x, p)$ for every $p$. When we consider systems which do not depend on parameters, we will simply omit the notation $p$ (e.g. $f(x)$, $J(x)$). If $x^\ast$ is a stable fixed point of $f(x)$, we assume that the eigenvectors of $J(x^\ast)$ are linearly independent (i.e., $J(x^\ast)$ is diagonalizable). We denote the eigenvalues of $J(x^\ast)$ by $\lambda_i$ for $i =1, \dots, k$ and assume that they have multiplicities $\mu_i$ such that $\sum_{i=1}^k\mu_i = n$. 

\emph{Koopman Operator.} We limit our study of the Koopman operator to a basin of attraction $\cB(x^\ast)$ of an attractive fixed point $x^\ast$.  

\begin{defn} The \emph{basin of attraction} $\cB(x^\ast)$ of an attractive fixed point $x^\ast$ for the system $\dot x = f(x)$ is the set of initial conditions $x$ such that the flow $\phi_f(t,x)$ converges to $x^\ast$, i.e. $\cB(x^\ast) = \left\{ x\in \R^n\Bigl| \lim\limits_{t\rightarrow\infty}\phi_f(t,x) = x^\ast \right\}$. \hfill $\diamond$
\end{defn}

When the fixed point $x^\ast$ can be understood from the context, we write simply $\cB$. We will assume that $x^\ast$ is a stable hyperbolic fixed point, that is, the eigenvalues $\lambda_j$ of the Jacobian matrix $J(x^\ast)$ are such that $\Re(\lambda_j) <0$ for all $j$.

\begin{defn}
\emph{The Koopman semigroup of operators} associated with $\dot x = f(x)$ is defined as
\begin{gather}
U^t g(x) = g(\phi_f(t,x)), 
\end{gather}
where the functions $g:\R^n\rightarrow \C$ are called observables, $x\in\cB(x^\ast)$ and $\phi_f(t,x)$ is a solution to $\dot x = f(x)$.	\hfill $\diamond$
\end{defn}

The Koopman semigroup is linear~\cite{mezic2005}, so that it is natural to study its spectral properties. Consider the system with $f\in C^2$ on the basin of attraction $\cB(x^\ast)$ of a stable hyperbolic fixed point. We define \emph{the Koopman eigenfunctions} as a nontrivial function $s_j$ satisfying $U^t s(x) = s(\phi_f(t, x)) = s(x) \, e^{\lambda t}$, where $e^{\lambda t}$ belongs to the point spectrum of $U^t$ and we refer to such $\lambda\in \C$ as \emph{Koopman eigenvalues}. In particular, the eigenvalues $\lambda_j$ of the Jacobian matrix $J(x^\ast)$ are Koopman eigenvalues under the assumptions above. Furthermore, it can be shown that there exist $n$ eigenfunctions $s_j \in C^1(\cB)$ associated with eigenvalues $\lambda_j$ (see e.g. \cite{mauroy2014global}), and:
\begin{gather}
(f(x))^T \nabla s_j(x) = \lambda_j s_j(x), \label{eq:s-one}
\end{gather} 
If $f$ is analytic and if the eigenvalues $\lambda_j$ are simple (i.e., $\mu_j = 1$ for all $j$), then the flow of the system can be expressed (at least locally) through the following expansion (cf.~\cite{mauroy2013isostables}):
\begin{align}
\label{eq:expansion}&\phi_f(t, x) = x^\ast + \sum\limits_{j=1}^{n} s_j(x) v_j e^{\lambda_j t} + \\
&\sum\limits_{\begin{smallmatrix}
k_1,\dots,k_n\in\N_0\\
k_1+\dots+k_n>1 
\end{smallmatrix}} v_{k_1,\dots,k_n} \, s_1^{k_1}(x) \cdots s_n^{k_n}(x) e^{(k_1\lambda_1 +\dots k_n\lambda_l) t}, \nonumber
\end{align}
where $\N_0$ is the set of nonnegative integers, $v_j$ are the right eigenvectors of $J(x^\ast)$ corresponding to $\lambda_j$, the vectors $v_{k_1,\dots,k_n} \,$ are the so-called Koopman modes (see~\cite{mezic2005, mauroy2014global} for more details). We also note that it is implicitly assumed in~\eqref{eq:expansion} that $v_i^T \nabla s_j(x^\ast)  = 0$ for all $i\ne j$ and $ v_i^T \nabla s_i(x^\ast) = 1$. In the case of a linear system $\dot x = A x$ where the matrix $A$ has the left eigenvectors $w_i$, the eigenfunctions $s_i(x)$ are equal to $w_i^T x$ and the expansion~\eqref{eq:expansion} has only the finite sum (i.e., $v_{k_1,\dots,k_n} = 0$). A similar (but lengthy) expansion can be obtained if the eigenvalues $\lambda_j$ are not simple and have linearly dependent eigenvectors (see e.g.,~\cite{mezic2015applications}).

Let $\lambda_j$ be such that $0>\Re(\lambda_1)> \Re(\lambda_j)$, $j\neq 1$, then the eigenfunction $s_1$, which we call \emph{dominant}, can be computed using the so-called Laplace average~\cite{mauroy2013isostables}:
\begin{gather}\label{laplace-average}
g_{\lambda}^\ast(x) = \lim\limits_{t\rightarrow \infty}\frac{1}{t}\int\limits_0^t (g\circ \phi_f(s, x)) e^{-\lambda s} d s.
\end{gather}
For all $g\in C^1$ that satisfy $g(x^\ast)=0$ and $v_1^T \nabla g(x^\ast) \neq 0$, the Laplace average $g_{\lambda_1}^\ast$ is equal to $s_1$ up to a multiplication with a scalar. Note that we do not require the knowledge of $\cB(x^\ast)$ in order to compute $s_1$, since the limit in~\eqref{laplace-average} does not converge to a finite value for $x\not\in\cB(x^\ast)$. The eigenfunctions $s_j$ with $j\ge 2$ (non-dominant eigenfunctions) are generally harder to compute and are not considered in the present study.

The eigenfunction $s_1$ captures the asymptotic behavior of the system. In order to support this statement we consider the following definition.
\begin{defn} Let $s_1$ be a $C^1(\cB)$ eigenfunction corresponding to $\lambda_1$ such that $0> \Re(\lambda_1) > \Re(\lambda_j)$ for $j\ge 2$. The \emph{isostables} $\partial \cB_\alpha$ for $\alpha\ge 0$ are boundaries of the sublevel sets $\cB_\alpha =  \{x \in \cB | |s_1(x)| \le \alpha \}$, that is, 
$\partial\cB_\alpha=\left\{x \in \cB | |s_1(x)|=\alpha\right\}$.
\hfill $\diamond$
\end{defn}

In the case of a linear system $\dot{x}=A x$, the isostables are the level sets of $|s_1(x)|=|w_1^T x|$, where $w_1$ is the left eigenvector of $A$ associated with the dominant eigenvalue $\lambda_1$. If $\lambda_1$ is real, then $s_1$ is real and we will use the following notation  $\partial_+ \cB_\alpha=\left\{x\in \cB\Bigl| s_1(x) = \alpha\right\}$, $\partial_- \cB_\alpha=\left\{x\in \cB\Bigl| s_1(x) = -\alpha\right\}$ for $\alpha\ge 0$. Furthermore, when $\lambda_1$ is real and simple, it follows from~\eqref{eq:expansion} that the trajectories starting from the isostable $\partial\cB_\alpha$ share the same asymptotic evolution
$\phi_f(t,x) \rightarrow x^\ast + v_1 \, \alpha  e^{\lambda_1 t}$ with $t\rightarrow \infty$.
This implies that the isostables contain the initial conditions of trajectories that converge synchronously toward the fixed point. In particular, trajectories starting from the same isostable $\partial \cB_{\alpha_1}$ reach other isostables $\partial \cB_{\alpha_2}$ (with $\alpha_2 < \alpha_1$) after a time $\cT = \ln \left({\alpha_1} / {\alpha_2}\right) / |\Re(\lambda_1)|$. If $\lambda_1$ is not simple, then the isostables are not unique. However, we will choose specific isostables in the case of monotone systems (see Section~\ref{s:basins}). 

It can be shown that $\cB_\alpha=\cB$ as $\alpha\rightarrow\infty$, so that we will also use the notations $\cB_\infty$ and $\partial \cB_\infty$ to denote the basin of attraction and its boundary, respectively. More information about the isostables and their general definition using the flows of the system can be found in~\cite{mauroy2013isostables}.

Using the same tools as in~\cite{mauroy2014global}, it can be shown (provided that the eigenvectors of $J(x^\ast)$ are linearly independent) that the function $W_\beta(x) = \left(\sum\limits_{i = 1}^n \beta_i |s_i(x)|^p\right)^{1/p}$, with $p\ge 1$ and $\beta_i>0$, is a Lyapunov function, that is $W_\beta\in C^1(\cB)$, $W_\beta(x^\ast)=0$, $W_\beta(x)>0$, and $\dot {W}_\beta(x) < 0$ for all $x\in\cB\backslash x^\ast$. The properties $W_\beta(x^\ast)=0$ and $W_\beta(x)>0$ for $x\ne x^\ast$ stem from the fact that the zero level sets of $s_i$'s intersect only in the fixed point (see~\cite{mauroy2014global}). On the other hand, one can show that $\dot{W}_\beta \leq - \Re\{\lambda_1\} \|x\|_p$ by direct computation and definition of $s_i$'s.  In~\cite{mauroy2014global}, it was also discussed that $W_\beta$ can be used to estimate a basin of attraction of the system if all $s_i$ can be computed, since the function becomes infinite on the boundary $\partial \cB$ of the basin of attraction. This property is reminiscent of the definition of the \emph{maximal Lyapunov function}~\cite{vannelli1985maximal}. 
\begin{defn}
	A function $V_m : \R^n \rightarrow \R \cup \{ +\infty \}$ is called a \emph{maximal Lyapunov function} for the system $\dot x = f(x)$ admitting an asymptotically stable fixed point $x^\ast$ with a basin of attraction $\cB$, if 
	\begin{enumerate}
		\item $V_m(x^*) = 0$, $V_m(x) > 0$ for all $x\in\cB\backslash x^*$;
		\item $V_m(x) < \infty$ if and only if $x\in\cB$;
		\item $V_m(x) \rightarrow \infty$ as $x\rightarrow\partial\cB$ and/or $\|x\|\rightarrow +\infty$;
		\item $\dot V_m$ is well defined for all $x\in\cB$, and $\dot V_m(x)$ is negative definite for all $x\in \cB\backslash x^\ast$. \hfill $\diamond$
	\end{enumerate}
\end{defn}

In particular, it is straightforward to show that
\begin{gather*}
V(x) = \begin{cases}
W_\beta(x) & x \in\cB \\
\infty & \text{otherwise}
\end{cases}
\end{gather*}
is a maximal Lyapunov function and $V\in C^1(\cB)$ provided that the fixed point $x^\ast$ is stable and hyperbolic, and $f\in C^2(\cB)$. 

\emph{Partial Orders and Monotone Systems}. We define a partial order as follows: $x\succeq y$ if and only if $ x - y \in \Rnn^n$. In other words, $x\succeq y$ means that $x$ is larger or equal to $y$ entrywise. We write $x\not \succeq y$ if the relation $x \succeq y$ does not hold. We will also write $x\succ y$ if $x\succeq y$ and $x\ne y$, and $x\gg y$ if $x- y \in \Rp^n$. Note that cones $\cK$ more general than $\Rp^n$ can also be used to define partial orders~\cite{angeli2003monotone}, but unless stated otherwise we will consider $\cK=\Rnn^n$. 
Systems whose flows preserve a partial order relation are called \emph{monotone systems}. 

\begin{defn}\label{def:mon}
The system is \emph{monotone} on $\cD\times \cP$ if $\phi_f(t, x, p)\preceq \phi_f(t, y, q)$ for all $t\ge 0$ and for all $x\preceq y$, $p\preceq q$, where $x$, $y\in\cD$, $p$, $q\in\cP$. The system is \emph{strongly monotone} on $\cD\times \cP$, if it is monotone and if $\phi_f(t, x, p) \ll \phi_f(t, y, q)$ holds for all $t>0$ provided that $x\preceq y$, $p\preceq q$, and either $x\prec y$ or $p\prec q$ holds, where $x$, $y\in\cD$ and $p$, $q\in\cP$. \hfill $\diamond$
\end{defn}

A certificate for monotonicity is given by the \emph{Kamke-M\"uller} conditions (see e.g~\cite{angeli2003monotone}). The certificate amounts to checking the sign pattern of the matrices $\partial f(x, p)/\partial x$ and $\partial f(x, p)/\partial p$. We will also consider the comparison principle~\cite{hirsch2005monotone}, which is typically used to extend some properties of monotone systems to a class of non-monotone ones. 

\begin{prop}
\label{lem:comp-prin-control}
 Consider the dynamical systems $\dot x = f(x)$ and $\dot x = g(x)$. Let one of the two systems be monotone on $\cD$. If $g(x) \succeq  f(x)$ for all $x\in\cD$ then $\phi_{g}(t, x_2) \succeq  \phi_{f}(t, x_1)$ for all $t\ge 0$ and for all $x_2\succeq  x_1$. \hfill $\diamond$
\end{prop}

\section{Increasing Functions and Order-Convex Sets}\label{s:order-convex}
Our subsequent derivations are based on the properties of the \emph{increasing} functions in $\R^n$, which were studied in the context of partial orders, for example, in~\cite{shaked2007stochastic}.
\begin{defn}
We call the set $\dom(g) = \{x  \in\R^n | |g(x)| < \infty  \}$ \emph{the effective domain} of a function $g$. A function $g:\R^n \rightarrow \R\cup \{-\infty, +\infty\}$ is called \emph{increasing} with respect to the cone $\Rnn^n$ if $g(x) \ge g(z)$ for all $x\succeq z$ and $x$, $z\in \dom(g)$. \hfill $\diamond$ 
\end{defn}

In this section, we study the properties of the \emph{sublevel} sets $\cA_\alpha=\left\{ x\in\R^n \Bigl| |g(x)| \le \alpha\right\}$ of increasing functions continuous on their effective domain $\dom(g)$. We first introduce a few concepts. Let open and closed intervals induced by the cone $\Rnn^n$ be defined as $[[x,~y]] = \{ z\in\R^n | x\ll  z \ll  y \}$ and $[x,~y] = \{ z\in\R^n | x\preceq  z \preceq  y \}$, respectively. 

\begin{defn}
A set $\cA$ is called \emph{order-convex} if, for all $x$, $y \in \cA$, the closed interval $[x,~y]$ is a subset of $\cA$. \hfill $\diamond$
\end{defn}

For an order-convex set $\cA$, the set of \emph{maximal elements $\partial_+ \cA$} (respectively, the set of \emph{minimal elements $\partial_- \cA$}) of $\cA$ is a subset of the boundary $\partial \cA$ of $\cA$ such that if $y\ll  z$ for some $y\in\partial_+ \cA$ (resp., if $y\gg  z$ for some $y\in\partial_- \cA$), then $z\not\in\cA$. It follows from the definition that for all $y,z \in \partial_+ \cA$ (or $y,z \in \partial_- \cA$), we cannot have $y\gg  z$ or $z\gg  y$. We have the following proposition, which is similar to results from~\cite{franklin1962orderconvex}.

\begin{prop}\label{prop:order-convex-general}
Let $\cA\in\R^n$ be order-convex. Then

(i) the boundary $\partial \cA$ of $\cA$ is the union of $\partial_+\cA$ and $\partial_- \cA$;

(ii) the interior of the set $\cA$ is the union of open intervals $[[x,~y]]$ over all $x\in\partial_- \cA$, $y\in\partial_+\cA$:
\begin{equation}\label{eq:order-convex-bound}
\inter(\cA) =\bigcup_{x\in\partial_- \cA,\,\,y\in\partial_+ \cA} [[x,~y]].
\end{equation}
\end{prop}
\begin{proof}
(i) It follows from their definition that the sets $\partial_- \cA$ and $\partial_+\cA$ are the subsets of the boundary $\partial \cA$, if they are not empty. Hence, we only need to show that $\partial \cA \subseteq \partial_-\cA \cup \partial_+\cA$. If for $z\in\partial\cA$ there exist $x\in\partial_-\cA$ and $y\in \partial_+\cA$ such that $z \in [[x,~y]]$, then $z\in\inter(\cA)$. 
If there exists $z\in\partial\cA$ such that there exists no $x\in\cA$ with $x \preceq z$ or $x\succeq z$, then $z$ itself is the maximal or the minimal element of $\cA$. 

(ii) If for some point $z$ in $\cA$ there does not exist a minimal element $x\in\partial_- \cA$ and a maximal element $y\in\partial_+\cA$ such that $x \ll z \ll y$, then $z\in\partial\cA$ and $z$ cannot belong to the interior of $\cA$. Hence for all $z\in\inter(\cA)$, there exist $x\in\partial_- \cA$ and $y\in\partial_+\cA$ such that $z\in[[x,~y]]$, which proves the claim. 
\end{proof}

Now we discuss the connection between order-convex sets and connected sets. Recall that the set $\cA\subset\R^n$ is called \emph{connected} if for any two points $x$, $y \in \cA$, there exists a path $\gamma(t)$ (i.e., a continuous curve $\gamma:[0,1]\rightarrow\R^n$) with $\gamma(0) = x$, $\gamma(1) = y$, and such that $\gamma(t)\in \cA$ for all $t\in[0,1]$. The set is called \emph{simply connected} if it is connected and if every path between $x$, $y \in \cA$ can be continuously transformed, staying within $\cA$, into any other such path while preserving the endpoints. Since a union of sets can be disconnected, Proposition~\ref{prop:order-convex-general} does not imply that order-convex sets are simply connected or even connected. However, order-convex and connected sets are related to sublevel sets $\cA_\alpha=\left\{ x\in\R^n \Bigl| |g(x)| \le \alpha\right\}$ of increasing functions. 

\begin{prop} \label{prop:order-convex}
Let $g:\cD \to \R$ be a continuous function, where $\cD$ is an open order-convex set. Then: 

(i) the function $g$ is increasing with respect to $\Rnn^n$ if and only if the sublevel sets $\cA_\alpha \subseteq \cD$ are order-convex and connected for any $\alpha \ge 0$;

(ii) if $\cD \subseteq \R^2$ and if the function $g$ is increasing with respect to $\Rnn^2$, then the sets $\cA_\alpha \subseteq \cD$ are simply connected for any $\alpha \ge 0$. \hfill $\diamond$
\end{prop}
\begin{proof} 
(i) \emph{Sufficiency.} Connectivity of $\cA_\alpha$ directly follows from the continuity of $g$ (see e.g. Proposition 1 in \cite{tibilelli1994connectedness}). Let $x$, $y\in \cA_\alpha$ for some $\alpha\ge 0$ and $z$ be such that $x\preceq z \preceq y$. Since $\cD$ is order-convex, $z\in \cD$ so that the value $g(z)$ is defined and we have that $g(x)  \le g(z) \le g(y)$. Clearly, if $|g(x)|\le \alpha$ and $|g(y)|\le \alpha$, then $|g(z)| \le \alpha$ and $z\in \cA_\alpha$. This implies that the set $\cA_\alpha$ is order-convex.

\emph{Necessity.} Assume that $x \prec y$ with $x,y \in \cD$ and $g(x)>g(y)$. Consider the different cases:
\begin{enumerate}

\item Case $g(y) <0$ and $|g(x)| < |g(y)|$. We have $\cA_{|g(x)|} \subset \cA_{|g(y)|}$, so that $x \in \cA_{|g(y)|}$ and $y \in \partial_- \cA_{|g(y)|}$. Since $x \in \cA_{|g(y)|}$ and $y \in \partial_- \cA_{|g(y)|}$, we cannot have $x \prec y$ and we arrive at a contradiction.
\item Case $g(y) <0$ and $|g(x)| > |g(y)|$. We have $\cA_{|g(x)|} \supset \cA_{|g(y)|}$, so that $y \in \cA_{|g(x)|}$. This implies $g(x)>0$ and $x \in \partial_+ \cA_{|g(x)|}$. Since $y \in \cA_{|g(x)|}$ and $x \in \partial_+ \cA_{|g(x)|}$, we cannot have $x \prec y$ and we arrive at a contradiction.
\item Case $g(y) \geq 0$. Since $g(x) > g(y) \ge 0$, this case is treated in the same way as the case (2).
\end{enumerate}
It follows that $x \preceq y$ implies $g(x) \leq g(y)$. 

(ii) Consider an arbitrary closed curve $\eta:[0,~1]\rightarrow\R^n$ in $\cA_\alpha$ and the set $\cA^\eta$ enclosed by the curve $\eta$. For all $z \in \cA^\eta$, there exist points $y^1 \in \eta \subset \cA_\alpha$ and $y^2 \in \eta \subset \cA_\alpha$ such that $y^1 \prec z \prec y^2$. Then order-convexity of $\cA_\alpha$ implies that $z \in \cA_\alpha$ and it follows that $\cA^\eta\subseteq \cA_\alpha$. Since the curve $\eta$ is in $\R^2$ we can shrink this curve continuously to a point which belongs to the set $\cA_\alpha$. Since the curve is an arbitrary closed curve in $\cA_\alpha$, the set $\cA_\alpha$ is simply connected.
\end{proof}

Proposition~\ref{prop:order-convex} also holds for functions $g:\R^n\rightarrow \R\cup\{-\infty, +\infty\}$ continuous on their effective domain $\dom(g)$. In this case, the  function $g$ is increasing if and only if $\dom(g)$ and $\cA_\alpha$ are connected and order-convex. It also stands to reason that point (ii) may potentially be extended to the case of $\R^n$ due to the structure of order-convex sets. However, such an extension has proved to be cumbersome, therefore, we leave it outside the scope of this paper. 

\section{Geometric Properties of Monotone systems} \label{s:basins}
\subsection{Spectral Properties and Lyapunov Functions}  
\label{s:spec-mon}

We first establish the spectral properties of the Koopman operator associated with monotone systems.
\begin{prop} \label{prop:mon-eig-fun}
	Consider the system $\dot x = f(x)$ with a stable hyperbolic fixed point $x^\ast$ with a basin of attraction $\cB$. Assume that  $\Re(\lambda_1) > \Re(\lambda_j)$ for all $j\ge 2$.  Let $v_1$ be a right eigenvector of the Jacobian matrix $J(x^\ast)$ and let $s_1\in C^1(\cB)$ be an eigenfunction corresponding to $\lambda_1$ (with $v_1^T \nabla s_1(x^\ast) = 1$).
	
	(i) if the system is monotone on $\cB$, then $\lambda_1$ is real and negative, and there exist $s_1$ and $v_1$ such that $s_1(x) \ge s_1(y)$ for all $x$, $y\in \cB$ satisfying $x\succeq y$, and $v_1\succ 0$; 
	
	(ii) if the system is strongly monotone on $\cB$, then $\lambda_1$ is simple (i.e., its multiplicity $\mu_1 = 1$), real and negative, $s_1(x)> s_1(y)$ for all $x$, $y\in \cB$ satisfying $x\succ y$,  $v_1\gg 0$. \hfill $\diamond$
\end{prop}

The proof of Proposition~\ref{prop:mon-eig-fun} is almost identical to the proof of a similar result in~\cite{sootla2015koopman}, and hence it is omitted. In both cases, the conditions on $\lambda_1$, $v_1$ and $s_1$ are only necessary and not sufficient for monotonicity, which is consistent with the linear case and necessary conditions for positivity. We note that $s_1\in C^1(\cB)$, if for example $f\in C^2$. We note that since eigenfunctions define isostables, this proposition plays a crucial role in the derivation of geometric properties of monotone systems. Additionally, a maximal Lyapunov function can be constructed by using $s_1$.

\begin{prop} Assume that the system $\dot x = f(x)$ is monotone on the basin of attraction $\cB$ of a stable hyperbolic fixed point $x^\ast$. If $\partial \cB_0$ intersects $x^\ast + \Rnn^n$ only in $x^\ast$, where $+$ is the Minkowski addition, then $s_1$ is a Lyapunov function of $\dot x = f(x)$ on $X = x^\ast + \Rnn^n\bigcap \cB$. Moreover, the function $V: X\rightarrow \Rnn$, equal to $s_1$ on $X$, can be extended to $\R^n$ so that this extension is a maximal Lyapunov function. \hfill $\diamond$
\end{prop}

\begin{proof}
	It is clear that the set $x^\ast + \Rnn^n$ is forward-invariant for monotone systems. Moreover, according to~\eqref{eq:s-one}, we have  
	\begin{gather*}
	\frac{d s_1(x)}{dt} = (f(x))^T \nabla s_1(x) = \lambda_1 s_1(x),
	\end{gather*}
	where $\lambda_1$ is real and negative and $s_1(x)$ is equal to zero only for $x = x^\ast$ on $x^\ast + \Rnn^n$. Hence $s_1$ is a Lyapunov function on $X$. The second part of the statement is straightforward.
\end{proof}

The assumption on $\partial \cB_0$ is needed, since for some systems the isostable $\partial\cB_0$ can intersect $x^\ast + \Rnn^n$ along one of the axes. For planar systems, this may occur if, for example, $\partial f_1(x)/\partial x_2 = 0$ on $x^\ast + \Rnn^2\cap \cB$ (that is, the Jacobian matrix is reducible on $x^\ast + \Rnn^2\cap \cB$). On the other hand, $\partial\cB_0 \cap x^\ast + \Rnn^n =\{x^\ast\}$ if the system is strongly monotone. 

\subsection{Geometry of Basins of Attraction and Isostables}
\label{ss:geometry}
In this section, we study the properties of $\cB_\alpha$, including the limit case $\alpha = \infty$, where some of the properties of eigenfunctions may not longer hold. First, we recall the following result in monotone systems theory. A proof can be found for example in~\cite{sootla2015pulsesaut}.

\begin{prop} \label{prop:basin-proper}
Let the system $\dot{x} = f(x)$ be monotone on the basin of attraction $\cB(x^\ast)$ of an asymptotically stable fixed point $x^\ast$, then $\cB(x^\ast)$ is order-convex. \hfill $\diamond$
\end{prop}

Let $s_1$ be an eigenfunction corresponding to $\lambda_1$, which is increasing on $\cB(x^\ast)$. Since $\cB(x^\ast)$ is order-convex and connected, we can extend $s_1$ to $\R^n$ by assigning $\infty$ on $\R^n\backslash \cB(x^\ast)$. With a slight abuse of notation, we denote the extended function by $s_1$ as well. In general, for an eigenvalue $\lambda_1$ with a multiplicity $\mu_1$, there exist $\mu_1$ eigenfunctions corresponding to $\lambda_1$. Therefore, without loss of generality, we define the isostables as level sets of an increasing eigenfunction $s_1$ corresponding to $\lambda_1$.

The role of order-convexity and topology induced by order-intervals has been studied in the literature. Most of the results were shown for the semiflow case (e.g.,~\cite{jiang2004saddle}). We expand these results to characterize the geometric properties of sublevel sets of the eigenfunction $s_1$ as follows: 

\begin{prop} \label{thm:level-set-proper}
Let the system~\eqref{sys:f} have a stable hyperbolic fixed point $x^\ast$ with a domain of attraction $\cB(x^\ast)$. If the system is monotone on $\cB(x^\ast)$ with $s_1\in C^1(\cB(x^\ast))$, then

(i) the set $\cB_\alpha\subseteq \cB(x^\ast)$ is order-convex and connected for any nonnegative $\alpha$ including $\alpha = \infty$;

(ii) the boundary $\partial \cB_\alpha$ of $\cB_\alpha$ can be split into two manifolds: the set of minimal elements equal to $\partial_- \cB_\alpha$ and the set of maximal elements equal to $\partial_+ \cB_\alpha$. Moreover, the manifolds $\partial_- \cB_\alpha$, $\partial_+ \cB_\alpha$ do not contain points $x$, $y$ such that $x\gg  y$ for any nonnegative $\alpha$ including $\alpha = \infty$;

(iii) if the system is strongly monotone, then the manifolds $\partial_- \cB_\alpha$, $\partial_+ \cB_\alpha$ do not contain points $x$, $y$ such that $x\succ  y$ for any finite nonnegative $\alpha$.\hfill $\diamond$
\end{prop}
\begin{proof}
(i) By Proposition~\ref{prop:mon-eig-fun}, the eigenfunction $s_1$ is increasing, while its effective domain $\dom(s_1) = \cB(x^\ast)$ is order-convex by Proposition~\ref{prop:basin-proper}. Hence, the result follows from Proposition~\ref{prop:order-convex}. 

(ii) We have shown in the first point that the set $\cB_\alpha$ is order-convex for any nonnegative $\alpha$ including $+\infty$. Hence, the first statement follows by Proposition~\ref{prop:order-convex-general}. The second statement follows by definition of the set of minimal (maximal) elements and the fact that $\partial \cB_\alpha = \partial_- \cB_\alpha \cup \partial_+ \cB_\alpha$.

(iii) Let there exist $x$, $y \in \partial_- \cB_\alpha$ such that $x\succ  y$. Since $x$, $y \in \partial_- \cB_\alpha$, we have that $s_1(x) = s_1(y)$, but according to Proposition~\ref{prop:mon-eig-fun}, $x\succ  y$ implies that $s_1(x) > s_1(y)$. Hence no such $x$ and $y$ exist. Similarly, the case of $\partial_+ \cB_\alpha$ is shown.
\end{proof}

It is important to note that the boundary $\partial\cB$ can contain two points $x$, $y$ such that $x\gg  y$. But in this case these points belong to different manifolds $\partial_- \cB_\alpha$, $\partial_+ \cB_\alpha$. 

As we have shown above, some of the geometric properties of $\cB_\alpha$ are preserved in the limiting case $\alpha=\infty$. However, the third point of Proposition~\ref{thm:level-set-proper} is shown only for the case of a finite $\alpha$. This is due to the fact that the notion of order-convexity does not fully capture the properties of strictly increasing functions such as the dominant eigenfunction $s_1$ of a strongly monotone system. We discuss this issue under additional assumptions. We will consider the case of a bistable monotone system, which allows deriving many geometric properties of monotone systems. We make the following assumptions, which we also use in the sequel.
\begin{enumerate}
  \item[{A1.}] Let the system~$\dot x = f(x)$ have two asymptotically stable fixed points in $\cD_f$, denoted as $x^\ast$ and $x^\bullet$, and let $\cD_f=\cl(\cB(x^\ast)\cup \cB(x^\bullet))$;
  \item[{A2.}] Let the fixed points be such that $x^\bullet \succeq  x^\ast$.
\end{enumerate}

The following proposition is a direct corollary of the results in~\cite{jiang2004saddle}, where the basins of attraction of semiflows were studied. We prove it for completeness, in order to discuss how assumptions in~\cite{jiang2004saddle} translate into our simplified case.

\begin{prop} \label{prop:strong-mon-separ}Let the system satisfy Assumptions {A1}--{A2} and be strongly monotone on $\cD_f$. Assume also that for any bounded set $A\in\cD_f$ the set $O(A) = \bigcup_{t\ge 0}\phi(t, A)$ is bounded. Then the boundary between the basins $\cB(x^\ast)$ and $\cB(x^\bullet)$ does not contain two points $x$, $y$ such that $x \succ y$. \hfill $\diamond$
\end{prop}
\begin{proof}
This result is a corollary of Theorem~2.2 in~\cite{jiang2004saddle}, hence we need to make sure that all assumptions are satisfied. Assumption (A1) in~\cite{jiang2004saddle} states that the semiflow should be an $\alpha$-contraction, where $\alpha(\cdot)$ is a Kuratowski measure of non-compactness. Since our flow is in $\R^n$, $\alpha(B)$ for any bounded set is equal to zero. The strongly order preserving (SOP) property and strong monotonicity in our case are equivalent. Furthermore, the operator $\partial \phi(t, x)$, which is the fundamental solution of $\dot{\delta x} = J(\phi(t,x)) \delta x$, is strongly positive for strongly monotone systems. Finally we complete the proof by applying Theorem~2.2 in~\cite{jiang2004saddle}.
\end{proof}

The assumption on boundedness of $O(A)$ is technical and is typically made to avoid pathological cases in monotone systems theory. Hence the only assumption to check is strong monotonicity, which is valid if the system is monotone and the Jacobian is irreducible for all $x\in\cD_f$. Proposition~\ref{prop:strong-mon-separ} offers a strong theoretical result, however, its practical implications is limited in comparison with Proposition~\ref{thm:level-set-proper}. Even if we establish that the system is strongly monotone it does not offer any direct computational advantage in comparison with the monotone --but not strongly monotone-- case. 

\subsection{Basins of Attraction of Bistable Systems}\label{ss:basins}
In this subsection, we first consider the class of non-monotone systems with vector fields that are bounded from below and above by vector fields of monotone systems. We show that the basins of attraction of such systems can be bounded by basins of attraction of the bounding systems. 

\begin{thm} \label{thm:comp-sys-b}
Let the systems~$\dot x = f(x)$, $\dot x = h(x)$, and $\dot x = g(x)$ satisfy Assumptions~{A1}--{A2}. Let $\cD = \cD_g  = \cD_f = \cD_h$, the systems~$\dot x = h(x)$ and $\dot x = g(x)$ be monotone on $\cD$ and 
\begin{align}\label{ineq:comp-prin}
	g(x) \preceq  f(x)\preceq  h(x)\textrm{ for } x\in\cD.
\end{align}
Assume also that the fixed points $x^\ast_g$,  $x^\ast_f$, $x^\ast_h$, $x^\bullet_f$ satisfy 
\begin{gather}\label{eq:condss-b} 
x^\ast_g, x^\ast_f, x^\ast_h \in \cB(x^\ast_g)\cap \cB(x^\ast_h), \\ 
\label{eq:condss2-b} x^\bullet_f\not \in [x^\ast_g, x^\ast_h]. 
\end{gather}
Then the following relations hold:
\begin{equation}
\label{st:bound-b} \cB (x^\ast_{g}) \supseteq \cB (x^\ast_{f}) \supseteq \cB(x^\ast_{h}).
\end{equation} 
Moreover, the sets $\cB (x^\ast_{g})$, $\cB(x^\ast_{h})$ are order-convex.\hfill $\diamond$
\end{thm}
\begin{proof} 
i) First we note that the assumption~\eqref{eq:condss-b} implies that $x^\ast_g\preceq  x^\ast_f\preceq  x^\ast_h$. Indeed, 
by Proposition~\ref{lem:comp-prin-control} for all $t\ge 0$, we have
$\phi_g(t, x^\ast_f) \preceq  \phi_f(t, x^\ast_f) \preceq  \phi_h(t, x^\ast_f)$, and thus taking the limit $t\to\infty$ we get $x^\ast_g \preceq  x^\ast_f\preceq  x^\ast_h$.

ii) Next we show that $g(x)\preceq  f(x)$ for all $x \in\cD$ implies that $\cB (x^\ast_{g}) \supseteq \cB (x^\ast_{f})$. Let $y\in\cB (x^\ast_{f})$. By Proposition~\ref{lem:comp-prin-control} we have that $\phi_{g}(t, y) \preceq  \phi_{f}(t, y)$. Furthermore, $\lim_{t\rightarrow\infty}\phi_{f}(t, y) = x^\ast_f$ and $\phi_{g}(t, y)$ converges to $x^\ast_g$ or $x^\bullet_g$, which implies that $\lim_{t\rightarrow\infty}\phi_{g}(t, y)\succeq x^\ast_g$. Hence, there exists a time $T$ such that
\[
x^\ast_g - \varepsilon \bfone\ll  \phi_{g}(t, y) \ll  x^\ast_f+\varepsilon \bfone
\] 
for all $t>T$ and some positive $\varepsilon$. We can pick a small $\varepsilon$ such that $x^\ast_f+\varepsilon \bfone$ and $x^\ast_g -\varepsilon \bfone$ lie in $\cB(x^\ast_g)$ (due to~\eqref{eq:condss-b}). According to Proposition~\ref{prop:basin-proper}, the flow 
$\phi_{g}(t, y)$ lies in $\cB(x^\ast_g)$ and hence $y\in\cB(x^\ast_g)$, which completes the proof.

iii) Similarly to ii), we have that $\cB (x^\ast_{g}) \supseteq \cB(x^\ast_{h})$.

iv) Finally, we show that $\cB (x^\ast_{f}) \supseteq \cB(x^\ast_{h})$. Let $y\in\cB(x^\ast_{h})$. By Proposition~\ref{lem:comp-prin-control}, we have that 
\[
\phi_{g}(t, y )\preceq   \phi_{f}(t, y ) \preceq  \phi_{h}(t, y), 
\]
for all $t\ge 0$. Furthermore, due to iii), we have that $y\in\cB (x^\ast_{g})$ and $\phi_{g}(t, y )$ converges to $x^\ast_{g}$, and there exists a $T$ such that 
\[
x^\ast_g -\varepsilon \bfone\preceq   \phi_{f}(t, y) \preceq  x^\ast_h + \varepsilon \bfone
\] 
for all $t>T$ and some small positive $\varepsilon$. We can also choose an $\varepsilon$ such that $x^\ast_h + \varepsilon \bfone$ and $x^\ast_g -\varepsilon \bfone$ lie in $\cB(x^\ast_g)\cap \cB(x^\ast_h)$ due to~\eqref{eq:condss-b}. Hence the flow $\phi_{f}(t, y)$ belongs to the set $\{z | x^\ast_g - \varepsilon \bfone \preceq   z \preceq  x^\ast_h + \varepsilon \bfone\}$ for all $t > T$. Since the system~$\dot x = f(x)$ is bistable, the flow must converge to $x^\bullet_f$ or $x^\ast_f$. If the flow converges to $x^\bullet_f$, it violates condition~\eqref{eq:condss2-b}. Hence the flow $\phi_f(t,y)$ converges to $x^\ast_f$ and $y\in\cB(x^\ast_f)$.

v) Order-convexity of the sets $\cB(x^\ast_g)$ and $\cB(x^\ast_h)$ follows from Proposition~\ref{prop:basin-proper}.
\end{proof}

\begin{figure}[t]
\includegraphics[width=0.75\columnwidth]{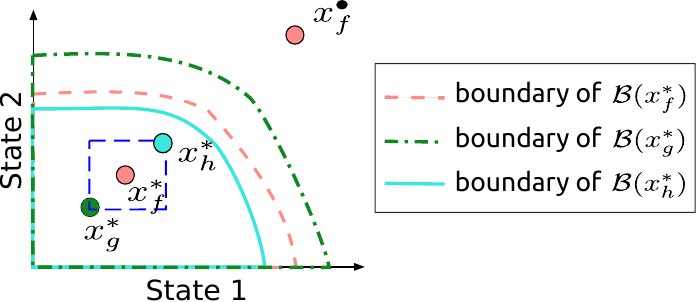}
\centering
\caption{A schematic depiction of the conditions~\eqref{eq:condss-b} and~\eqref{eq:condss2-b}. The condition~\eqref{eq:condss-b} ensures that all the fixed points lie in the intersection of the corresponding domains of attractions (in this case, it is $B(x^\ast_h)$). The fixed point $x^\bullet_f$ cannot lie in the dashed blue box due to condition~\eqref{eq:condss2-b}.}
\label{fig:propss}
\end{figure}

The conditions~\eqref{eq:condss-b},~\eqref{eq:condss2-b} are technical and generally easy to satisfy. An illustration of these conditions is provided in Figure~\ref{fig:propss}. Checking the condition~\eqref{eq:condss2-b} is equivalent to computing the stable fixed points. Similarly, condition~\eqref{eq:condss-b} holds if the trajectories of the systems~$\dot x = h(x)$ and $\dot x = g(x)$ initialized at $x^\ast_f$ converge to $x^\ast_g$ and $x^\ast_h$, respectively. This can be verified by numerical integration.

A key to using this result is the computation of bounding systems, which we discuss below. Systems with vector fields that can be bounded from above and below by vector fields of monotone systems are typically such that only a few interactions between the states are not consistent with monotonicity. These systems are called near-monotone and are often observed in biological applications (for a comprehensive discussion on near-monotonicity, see~\cite{sontag2007monotone}). Assume that there exists a single interaction which is not compatible with monotonicity. Namely, let the $(i,1)$-th entry in the Jacobian matrix $J(x)$ be smaller than zero for some $i\ne 1$ and some $x$. A bounding monotone system can be obtained by replacing the variable $x_1$ with a constant in the function $f_i(x_1,x_2,\dots, x_n)$, which removes the interaction between the states $x_i$ and $x_1$. Hence, we simply use $g_j(x) = f_j(x)$ if $i\ne j$ and $g_i(x) = f_i(\delta, x_2, \dots, x_n)$ with $\delta$ small enough. Similarly we set $h_j(x) = f_j(x)$ if $i\ne j$ and $h_i(x) = f_i(\eta, x_2, \dots, x_n)$ with $\eta$ large enough. More details on this procedure can be found in~\cite{sootla2015pulsesaut}. However, when the bounding systems are found, we still need to check the other conditions in the premise of Theorem \ref{thm:comp-sys-b}. It can happen that the bounding systems are not bistable, or not even stable (see, the toxin-antitoxin system in Section~\ref{s:examples-ta}). Unfortunately, we are not aware of an algorithm, which can guarantee bistability and monotonicity in the bounding systems.

A corollary from Theorem~\ref{thm:comp-sys-b} allows to estimate the basins of attraction of monotone systems under parametric uncertainty. 
\begin{cor}
\label{prop:basin-unc}
Consider a family of systems $\dot x = f(x, p)$ with a vector of parameters $p$ taking values from a set $\cP$. Let the systems $\dot x = f(x, p)$ satisfy Assumptions~{A1}--{A2} and be monotone on $\cD\times \cP$, where $\cD = \cD_{f(\cdot,  q)}$ for all $q\in\cP$. Consider an interval $[p_{\rm min}, p_{\rm max}]$ such that  
\begin{gather}
\label{eq:condssp-b} x^\ast(p_{\rm min})\in \cB(x^\ast(p_{\rm max})), x^\ast(p_{\rm max})\in \cB(x^\ast(p_{\rm min})),\\
\label{eq:condssp2-b} x^\bullet(p_{\rm min}) \not \in [x^\ast(p_{\rm min}), x^\ast(p_{\rm max})].
\end{gather}
Then the following relation holds: 
\begin{equation}\label{eq:boundsparam}
 \cB(x^\ast(p_{\rm min})) \supseteq \cB(x^\ast(p)) \supseteq \cB(x^\ast(p_{\rm max}))
\end{equation} 
for all $p$ in $[p_{\rm min}, p_{\rm max}]$. Moreover, the sets $\cB(x^\ast(p))$ are order-convex for all $p$ in $[p_{\rm min}, p_{\rm max}]$.
\end{cor}
\begin{proof}
Let $g(x)= f(x, p_{\rm min})$ and $h(x)= f(x,p_{\rm max})$. According to the order in the parameter set, we have that 
\begin{align}
g(x)\preceq  f(x, p) \preceq  h(x),
\end{align}
for all $(x,p)\in\cD \times [p_{\rm min}, p_{\rm max}]$. 

The interval $[x^\ast(p_{\rm min}), x^\ast(p_{\rm max})]$ is a subset of $\cB(x^\ast(p_{\rm max}))\bigcap \cB(x^\ast(p_{\rm min}))$, since the endpoints of this interval belong to this intersection according to~\eqref{eq:condssp-b}. Therefore for all $p\in [p_{\rm min}, p_{\rm max}]$ we have that $x^\ast(p) \in \cB(x^\ast(p_{\rm max}))\bigcap \cB(x^\ast(p_{\rm min}))$ and \eqref{eq:condss-b} in the premise of Theorem~\ref{thm:comp-sys-b} follows. 

Due to monotonicity we have that $x^\bullet(p_{\rm min}) \preceq  x^\bullet(p)$ for all $p_{\rm min} \preceq  p$. Since $x^\bullet(p_{\rm min})\succeq  x^\ast(p_{\rm min})$, the condition $x^\bullet(p) \in [x^\ast(p_{\rm min}), x^\ast(p_{\rm max})]$ implies that $x^\bullet(p_{\rm min}) \in [x^\ast(p_{\rm min}), x^\ast(p_{\rm max})]$  and contradicts~\eqref{eq:condssp2-b}. Hence $x^\bullet(p)\not \in [x^\ast(p_{\rm min}), x^\ast(p_{\rm max})]$ and~\eqref{eq:condss2-b} in the premise of Theorem~\ref{thm:comp-sys-b} follows, application of which completes the proof.
\end{proof}

Corollary~\ref{prop:basin-unc} implies that we can predict a direction of change in the basins of attraction subject to parameter variations if we check a couple of simple conditions. This result may be valuable for design purposes in some applications. For instance, in the toxin-antitoxin example (see Section~\ref{s:examples-ta}) it is desirable to make the basin of attraction of the fixed point corresponding to the high toxin concentration smaller. This would increase the likelihood of switching from this fixed point to another (which corresponds to low toxin concentrations) subject to intrinsic and/or exogenous noise. However, the considered toxin-antitoxin model is not monotone and further investigation is required to extend Corollary~\ref{prop:basin-unc} to a larger class of systems.

Another design problem is to determine the set of all parameters for which a monotone system is at least bistable. We offer the following development of this problem. 

\begin{prop}
\label{prop:bistability-region}
Consider a family of systems $\dot x = f(x, p)$ with a vector of parameters $p$ taking values from an order-convex set $\cP$. Let the system be strongly monotone on $\cD\times \cP$, where $\cD = \cD_{f(\cdot,  q)}$ for all $q\in\cP$, i.e. $\phi(t, x, p) \ll \phi(t, y, q)$ if $x\prec y$ or $p \prec q$. Consider an interval $[p_{\rm min}, p_{\rm max}]$ such that the systems $\dot x = f(x,p_{\rm min})$ and $\dot x = f(x,p_{\rm max})$ satisfy Assumptions~{A1}--{A2} and (\ref{eq:condssp-b},\ref{eq:condssp2-b}) hold. Then the intervals $[x^\ast(p_{\rm min}), x^\ast(p_{\rm max})]$, $[x^\bullet(p_{\rm min}), x^\bullet(p_{\rm max})]$ are compact attractors for  
the system $\dot x = f(x,p)$ for all $p\in [p_{\rm min}, p_{\rm max}]$. Furthermore, both intervals contain at least one fixed point such that the eigenvalues of the Jacobian at these fixed points have nonpositive real parts.\hfill $\diamond$
\end{prop}

\begin{proof} Since $\phi_f(t, x, p_{\rm min}) \preceq \phi_f(t, x, p) \preceq \phi_f(t, x, p_{\rm max})$ for all $p\in[p_{\rm min}, p_{\rm max}]$, the flow $\phi_f(t, x, p)$ with the initial condition $x \in [x^\ast(p_{\rm min})-\varepsilon \bfone, x^\ast(p_{\rm max})+\varepsilon \bfone]$ for a small $\varepsilon>0$ converges to $[x^\ast(p_{\rm min}), x^\ast(p_{\rm max})]$. The same result holds for the interval $[x^\bullet(p_{\rm min}), x^\bullet(p_{\rm max})]$. Since the intervals $[x^\ast(p_{\rm min}), x^\ast(p_{\rm max})]$ and $[x^\bullet(p_{\rm min}), x^\bullet(p_{\rm max})]$ do not intersect according to~\eqref{eq:condssp2-b}, they are compact attractors for the system $\dot x = f(x,p)$ for every $p$. According to Theorem~4.6 in~\cite{hirsch1985systems}, both intervals contain a fixed point such that the Jacobian at this fixed point has eigenvalues with nonpositive real parts. 
\end{proof}

We note that these fixed points are not necessarily stable in the classical sense. However, using strong monotonicity and the convergence criterion~\cite{hirsch2005monotone} we can show that $e = \lim\limits_{t\rightarrow\infty} \phi(t, x^\ast(p_{\rm min}), p)$ is a fixed point for all $p\in[p_{\rm min}, p_{\rm max}]$. Furthermore, all points in $[x^\ast(p_{\rm min}), e]$ are converging to $e$ due to strong monotonicity. Fixed points with this property are called stable from below~\cite{hirsch2005monotone}. 

\section{Discussion on Computation of Basins of Attraction and Isostables}
\label{s:alg-disc}
\subsection{Lyapunov Methods} 
In control theory, a go-to approach for computing forward-invariant sets (not only basins of attraction) of dynamical systems is sum-of-squares (SOS) programming, see e.g.~\cite{henrion2014convex} and~\cite{valmorbida2014region} and the references within. This approach can be applied to systems with polynomial vector fields.

Another option is to compute the eigenfunction $s_1$, which provides the isostables and the basin of attraction. In the case of a polynomial vector field, we can formulate the computation of $s_1$ as an infinite dimensional \emph{linear algebraic} problem using~\eqref{eq:s-one}. Hence, we can provide an approximation of $s_1$  using linear algebra by parameterizing $s_1$ with a finite number of basis functions~\cite{mauroy2014global}. On another hand, we can estimate $s_1$ \emph{directly from data} using dynamic mode decomposition methods~\cite{Schmid2010,Tu2014}. These two options provide extremely cheap estimates of $s_1$. In fact, the algebraic methods (as demonstrated in~\cite{mauroy2014global}) also provide good estimates on basins of attraction. However, we cannot typically compute estimates with an excellent approximation quality, which comes as a tradeoff for fast computations. The eigenfunction $s_1$ can also be computed on a mesh grid by using Laplace averages~\eqref{laplace-average} and by simulating a number of trajectories with initial points on this mesh grid. Interpolation or machine learning methods can then be applied to estimate the dominant eigenfunction. 

In~\cite{valmorbida2014region}, the authors combined the maximal Lyapunov function approach with SOS techniques, which resulted in an iterative semidefinite program. As we have discussed above, the function $s_1$ can be used to construct a maximal Lyapunov function. Hence the main difference between using $s_1$ and \cite{valmorbida2014region} is algorithmic. In our point of view, using $s_1$ can be beneficial, since we compute $s_1$ directly (by linear algebra or Laplace averages), while in~\cite{valmorbida2014region} it is required to optimize over the shape of the maximal Lyapunov function.

In~\cite{henrion2014convex}, the authors propose an approach that is conceptually similar to the computation of $s_1$. They also lift the problem to an infinite dimensional space, but in the framework of the so-called Perron-Frobenius operator acting on measures, which is dual to the Koopman operator. The authors propose an infinite-dimensional linear program to compute a specific occupation measure related to the basin of attraction. They then consider a series of relaxations using the moments of the measures to formulate their problem as an iterative finite-dimensional semidefinite program (as in the case of~\cite{valmorbida2014region}). In some sense, this is similar to the relaxation of the procedure to compute $s_1$ from \emph{an infinite-dimensional algebraic problem} to \emph{a finite-dimensional one}. We note that, in contrast to the computation of $s_1$, the methods proposed in~\cite{henrion2014convex} do not rely on the spectral properties of the operator.

Even though recent advances in optimization allowed solving some semidefinite programs as iterative linear programs~\cite{basis_pursuit}, SOS approaches still lead to highly dimensional iterative optimization algorithms. Therefore, the applicability of SOS-based methods to high dimensional systems is delicate due to memory and computational power requirements. In our opinion, there is a tradeoff between two options: computing rather cheaply $s_1$ (e.g., using Laplace averages) with weak guarantees or employing semidefinite programming with heavy computational requirements and strong guarantees. In this context, the estimation of $s_1$ is simply one possible option. 

\subsection{Data Sampling Algorithms}\label{s:alg}
In the case of monotone systems, it follows from Proposition~\ref{thm:level-set-proper} that for all $z_1,z_2 \in \cB_\alpha$ with $\alpha>0$ and $z_1 \ll z_2$, the set $[z_1, z_2]$ is an inner approximation of $\cB_\alpha$ (with a non-zero measure). With $\{z_i\}_{i =1}^{N}\in\cB_\alpha$, an inner approximation of $\cB_\alpha$ is computed as $\bigcup_{i, j = 1, \dots, N} [z_i,~z_j]$. It is also possible to build an outer approximation. Assume we want to compute an approximation of $\cB_\alpha$ on an interval $\B = [b^0,~b^1]$, where $b^0\in \cB_\alpha$ and $b^1\not \in\cB_\alpha$. Let $\{z_i^o\}_{i =1}^{N_o}$ be in $\B$, but not in $\cB_\alpha$. Then the outer approximation is computed as follows $\B / \bigcup_{i= 1, \dots, N_o} [z_i^o,~b_1]$. This implies that $\cB_\alpha$ can be estimated by a data sampled approach in the case of monotone systems. The data sampling algorithms presented below have a few of advantages over Lyapunov methods:
\begin{itemize}
	\item The algorithms can be applied to non-polynomial vector fields and have low memory requirements;
	\item The algorithms can be parallelized by generating several samples $z^j$ at the same time;
	\item Inner and outer approximations of $\cB_\alpha$ are computed at the same time;
	\item It is straightforward to compute estimates on $2$-D and $3$-D cross-sections of $\cB_\alpha$.
\end{itemize}
\begin{rem}
By $k$-D cross-sections we mean the following. We fix $n-k$ state components to be equal to a constant, that is $x_{j_i} = p_i$ for $i=1,\dots n-k$ and $j_i \in \cI \subset \{1,\dots,n\}$, and define a cross-section as follows:
 \begin{gather*}
 \widetilde \cB(z, p, \cI) = \left\{ x \in \cB(z) \Bigl| x_{j_i} = p_i, j_i \in \cI \right\},
 \end{gather*}
 where $z$ is an asymptotically stable fixed point. \hfill $\diamond$
\end{rem}

 We discuss the approach in an abstract form using oracles. Let $\cO(z):\B\rightarrow \{0,1\}$ be an increasing function, that is $\cO(x) \le \cO(y)$ for all $x\preceq y$ and let $\B = \{z \in\R^{n}| b^0 \preceq z \preceq b^1 \}$ be such that $\cO(b^0) = 0$ and $\cO(b^1) = 1$.	Our goal is to compute the set $\{z \in \B | \cO(z) = 0 \}$. For the computation of $\cB_\alpha$ we use the following oracle:
\begin{gather}
\label{eq:oracle}
\cO(z) = \begin{cases}
0  & \textrm{ if } |s_1(z)| < \alpha \textrm{ and } \|\phi(T,z) - x^\ast\| < \varepsilon, \\
1  & \textrm{ otherwise},
\end{cases} 
\end{gather}
where the value of the eigenfunction $s_1(z)$ is computed using Laplace averages~\eqref{laplace-average}, and $T$ is large enough. In our examples, we chose the observable $g(x) = w_1^T (x-x^\ast)$, where $w_1$ is the right dominant eigenvector of the Jacobian matrix $J(x^\ast)$. We need the second condition $\|\phi(T,z) - x^\ast\| < \varepsilon$ in \eqref{eq:oracle} to make sure that the points also lie in $\cB$. Even though it is unlikely to have a point $z$ with a finite $s_1(z)$ not lying in $\cB$, such situations can occur numerically. If we need to compute $\cB$, then we drop the condition $|s_1(z)| < \alpha$, since here $\alpha = \infty$.

The main idea of the algorithm is based on the increasing property of $\cO$. If a sample $z^j$ is such that $\cO(z^j) = 0$, then for all $w \preceq z^j$ we have $\cO(w)  = 0$. Similarly, if a sample $z^j$ is such that $\cO(z^j) = 1$, then for all $w \succeq z^j$ we have $\cO(w)  = 1$. Therefore, we need to keep track of the largest (in the order) samples $z^j$ with $\cO(z^j) =0$, the set of which we denote $\cM^{\rm min}$, and the smallest (in the order) samples $z^j$ with $\cO(z^j) = 1$, the set of which we denote $\cM^{\rm max}$. The set $\cM^{\rm min}$ lies in $\{z \in\R^{n} | \cO(z) = 0 \}$, while the set $\cM^{\rm max}$ lies in $\{z\in\R^{n} |\cO(z) = 1\}$. Since $\cO(z)$ is an increasing function, the set $\cM^{\rm min}$ (respectively, the set $\cM^{\rm max}$) can be used to build a piecewise constant inner approximation (respectively, an outer approximation) of the set $\{z \in \B | \cO(z) = 0 \}$. In order to improve the approximation quality at every step, we generate new samples $z$ in the set 
\begin{gather*}
\A = \left\{z\in\R^n | x \not\succ z, z\not \succ y, \,\,  \forall x \in \cM^{\rm min}, \forall y \in \cM^{\rm max} \right\},
\end{gather*}
and update $\cM^{\rm min}$ and $\cM^{\rm max}$. The approach is summarized in Algorithm~\ref{alg:s1}, the major parts of which are the stopping criterion and the generation of new samples. 

\begin{algorithm}[t]
	\caption{Computation of the level set of $\cO(z)$}
	\label{alg:s1}
	\begin{algorithmic}[1]
		\State {\bf Inputs:} Oracle $\cO$, the initial set $\B$
		\State {\bf Outputs:} The sets of points $\cM^{\rm min}$, $\cM^{\rm max}$.
		\State Set $\A = \B$, $\cM^{\rm min} = \{b^0\}$,  $\cM^{\rm max} = \{b^1\}$
		\While{stopping criterion is not satisfied}
		\State Generate a sample $z\in \tilde \A$, where $\tilde \A \subset \inter(\A)$
		\State If $\cO(z) = 0$, then add $z$ to $\cM^{\rm min}$
		\State If $\cO(z) = 1$, then add $z$ to $\cM^{\rm max}$
		\State Update $\A$ using $\cM^{\rm min}$ and $\cM^{\rm max}$
		\EndWhile
	\end{algorithmic}
\end{algorithm}
In~\cite{sootla2016basins} it was proposed to use an estimate on the volume of $\A$ to establish a stopping criterion. It was also proposed to generate part of the samples randomly using a distribution with the support on $\tilde \A = \inter(\A)$, and to generate the rest greedily by finding the areas of $\A$ which contain largest intervals. While the volume of $\A$ is large, random sampling helps to learn the shape of the function $\cO(z)$. As the algorithm progresses, the greedy sampling ensures faster convergence of inner and outer approximations. In this algorithm, it is required to keep the points in $\cM^{\rm min}$, $\cM^{\rm min}$ unordered, i.e., $x\not \succ y$ for all $x$, $y$ in $\cM^{\rm min}$ and $x \not \prec y$ for all $x$, $y\in\cM^{\rm max}$. Therefore after generation of new samples, we need to prune $\cM^{\rm min}$ (respectively, $\cM^{\rm max}$) by removing all $x$ for which there exists $y\in\cM^{\rm min}$ such that $x\prec y$ (respectively, $y\in\cM^{\rm min}$ such that $x\succ y$).
In~\cite{kim2016directed}, it was proposed to use 
\begin{gather*}
\tilde \A^\varepsilon = \left\{z\in\R^n | x \not\succ z, z\not \succ y, \,\,  \forall x \in \cM^{\rm min}_\varepsilon, \forall y \in \cM^{\rm max}_\varepsilon  \right\},
\end{gather*}
where $\varepsilon >0$ is called the learning rate, $\cM^{\rm max}_\varepsilon = \{y - \varepsilon \bfone \in \R^n| y \in \cM^{\rm max}\}$, $\cM^{\rm min}_\varepsilon = \{x + \varepsilon \bfone \in \R^n | x \in \cM^{\rm min}\}$. If $\tilde \A^\varepsilon$ is empty, then the learning rate is adjusted as follows: $\varepsilon = \alpha \varepsilon$ for some $\alpha \in (0,~1)$. The existence of $z$ in $\tilde \A^\varepsilon$ can be established using the solver~\cite{de2008z3}, which also produces a solution point $z$, if it exists. The stopping criterion is a lower bound on $\varepsilon$.

The algorithm from~\cite{sootla2016basins} allows controlling where the new samples are generated and ensures that the samples always decrease the volume of $\A$ in a maximal way according to the proposed heuristic. In order to do so, we need to sweep through a large number of points in $\cM^{\rm min}$, $\cM^{\rm max}$, which can be computationally expensive for a large $n$ (recall that $z\in\R^{n}$). At the same time, the algorithm from~\cite{kim2016directed} avoids computing an estimate of the volume of $\A$ by solving a feasibility problem using an efficient solver. Even though the algorithm~\cite{sootla2016basins} can potentially be implemented using efficient search methods over partially ordered sets, the algorithm~\cite{kim2016directed} seems to be more appealing due to the off-the-shelf feasibility solver. 

We finally note that in numerical experiments the algorithm from~\cite{sootla2016basins} exhibits exponential convergence in the estimation error versus the number of generated samples (at least in the cases $n =2$, $n=3$). In~\cite{kim2016directed} it is argued that, in order to converge to an $\varepsilon$ so that $\tilde \A^\varepsilon$ is empty, it is required to generate $\left(\frac{\max_i |b^1_i-b^0_i|}{\varepsilon}\right)^n$ samples in the worst case. Hence, this approach requires a finite number of samples to converge and potentially has an exponential convergence.

\section{Examples } \label{s:examples}
\subsection{A Two-State Toggle Switch} 
We illustrate our methods on a genetic toggle switch model (see e.g.~\cite{Gardner00}). We choose this model since it is extensively studied in synthetic biology so that our results can be verified with other techniques. We consider the following model:
\begin{gather}\label{sys:toggle-2d}
\begin{aligned}
\dot x_1 &= p_{1 1} + \frac{p_{1 2}}{ 1 + x_2^{p_{1 3}}} - p_{1 4} x_1, \\
\dot x_2 &= p_{2 1} + \frac{p_{2 2}}{ 1 + x_1^{p_{2 3}}} - p_{2 4} x_2,
\end{aligned}
\end{gather}
where all $p_{i j} \ge 0$. The states $x_{i}$ represent the concentration of proteins, whose mutual repression is modeled via a rational function. The parameters $p_{1 1}$ and $p_{2 1}$ model the basal synthesis rate of each protein. 
The parameters $p_{1 4}$ and $p_{2 4}$ are degradation rates, and $p_{1 2}$, $p_{2 2}$ describe the strength of mutual repression. The parameters $p_{1 3}$, $p_{2 3}$ are called Hill coefficients. The model is monotone on $\Rnn^2$ for all nonnegative parameter values with respect to the orthant $\diag{\begin{pmatrix} 1 & -1 \end{pmatrix}} \Rnn^2$. Moreover, the model is monotone with respect to all parameters but $p_{1 3}$, $p_{2 3}$. In this setting the fixed point $x^\ast$ has the state $x_2$ ``switched on'' ($x_2^\ast$ is much larger than $x_1^\ast$), while $x^\bullet$ has the state $x_1$ ``switched on'' ($x_1^\bullet$ is much larger than $x_2^\bullet$).

First, we check Corollary~\ref{prop:basin-unc} by considering
\begin{gather*}
 p = \begin{pmatrix}
  q_1 & q_2 & 4 & 1 \\ q_3 & q_4 & 3 & 2
 \end{pmatrix},
\end{gather*}
with the set of admissible parameters $\cQ = \{ q | q_{\rm max} \succeq_q q \succeq_q  q_{\rm min}\}$, where $q_{\rm min} = \begin{pmatrix}
	1.8 & 950 &  1.2 & 1050 
\end{pmatrix}$, $q_{\rm max} = \begin{pmatrix}
	2.2 & 1100 & 0.7 & 900 
\end{pmatrix}$, and $\succeq_q$ is induced by $\diag{\begin{pmatrix} 1 & 1 & -1 & -1\end{pmatrix}}\Rnn^4$. We compute the isostables $\partial \cB_\alpha(x^\ast(\cdot))$ with $\alpha = 0$, $2\cdot 10^3$, $\infty$ (where $\partial \cB_\infty(x^\ast(\cdot))$ is the boundary of the basin of attraction) for systems with parameters $q_{\rm min}$, $q_{\rm max}$, and $ q_{\rm int} = \begin{pmatrix}
2 & 1000 & 1 & 1000 
\end{pmatrix}$.

\begin{figure}[t]\centering
  \includegraphics[width = 0.8\columnwidth]{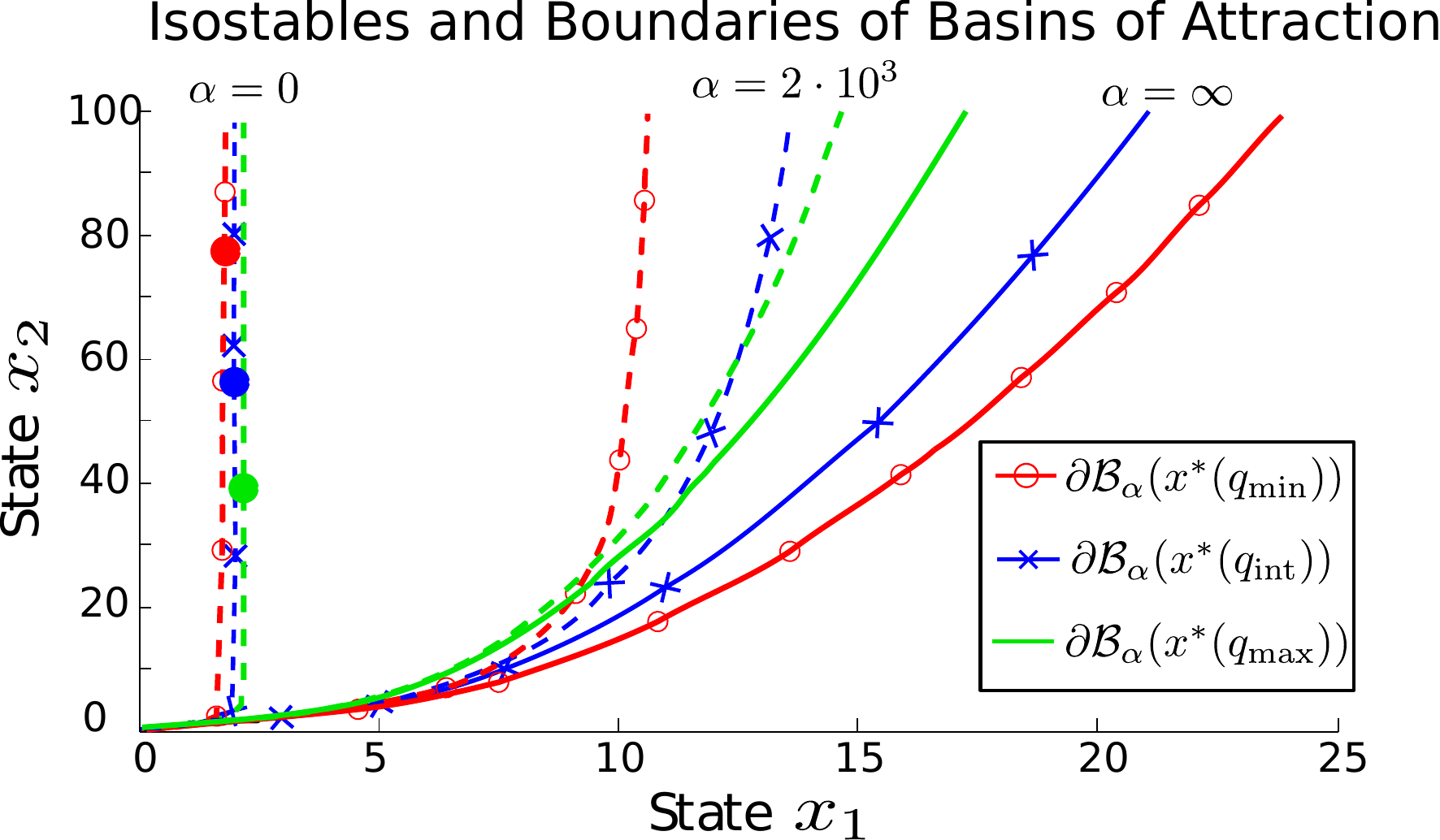}
\caption{Isostables $\partial \cB_{0}(x^\ast(q))$, $\partial \cB_{2\cdot 10^3}(x^\ast(q))$, and boundaries of basins of attraction $\partial \cB(x^\ast(q))$ for $q$ equal to $q_{\rm int}$, $q_{\rm min}$, or $q_{\rm max}$. The dots represent the fixed points $x^\ast(q_{\rm int})$ (blue), $x^\ast(q_{\rm min})$ (red), $x^\ast(q_{\rm max})$ (green) for different parameter values} \label{fig:toggle-2d-iso}
\end{figure}
 The computational results depicted in Figure~\ref{fig:toggle-2d-iso} suggest that for all parameter values $q\in\cQ$ the manifold $\partial \cB_{\infty}(x^\ast(q))$ will lie between the manifolds $\partial \cB_{\infty}(x^\ast(q_{\rm min}))$ and $\partial \cB_{\infty}(x^\ast(q_{\rm max}))$. It appears that $\partial \cB_0(x^\ast(q_{\rm int}))$ also lies between the manifolds $\partial \cB_0(x^\ast(q_{\rm min}))$ and $\partial \cB_0(x^\ast(q_{\rm max}))$, however, in a different order. This change of order and continuity of $s_1$ implies that there exists an $\alpha$ such that at least two manifolds $\partial \cB_\alpha(x^\ast(q_{\rm int}))$, $\partial \cB_\alpha(x^\ast(q_{\rm max}))$, $\partial \cB_\alpha(x^\ast(q_{\rm min}))$ intersect. This case is also depicted with $\alpha = 2\cdot 10^3$. This observation implies that $s_1(x,q)$ is not an increasing function in $q$. This is consistent with the linear case, where changes in the drift matrix $A$, will simply rotate the hyperplane $w_1^T x$, where $w_1$ is the left dominant eigenvector of $A$.
 
Another feature for a successful design of a genetic toggle switch is choosing the Hill coefficients. Computing the derivative of the repression term with respect to a Hill coefficient $p_{i 3}$ gives $-\dfrac{ p_{1 2}\ln(x) x^{p_{i 3}}}{(1 + x^{p_{i 3})^2}}$. Therefore for positive $x$ the derivative changes sign at $x = 1$, and the partial order with respect to $p_{i 3}$ cannot be defined on the whole state-space $\Rp^{2}$ and parameter space $\Rp^{2}$. Nevertheless, we can study how basins of attraction change subject to changes in Hill coefficients. Consider the following parameter values
\begin{gather}\label{param-values-coop}
 p = \begin{pmatrix}
 1 & 1000 & q_1 & 1 \\ 1 & 1000 & q_2 & 2
 \end{pmatrix}.
\end{gather}
The computation results are depicted in Figure~\ref{fig:toggle-2d-basins-coop}. We observe that the variations in basins of attraction are consistent with the result of Corollary~\ref{prop:basin-unc}, and the changes in Hill coefficients with respect to the order $\diag(1,~-1)\R^2$. It remains to verify if this result is an evidence of a partial order in Hill coefficients or simply a coincidence. However, in any case this result holds when we vary other parameters and hence can serve as a rule of thumb in designing toggle switches. We finally note that the curve $\partial\cB$ with $q_1 = 5$, $q_2 = 2$ is not convex or concave, which in other examples seems to be the case. We verified this observation by computing the curves with higher accuracy and smaller resolutions. The curve $\partial\cB$ with $q_1 = 5$, $q_2 = 2$ still separates $\R^2$ into two order-convex regions.
\begin{figure}[t]\centering
	\includegraphics[width = .6\columnwidth]{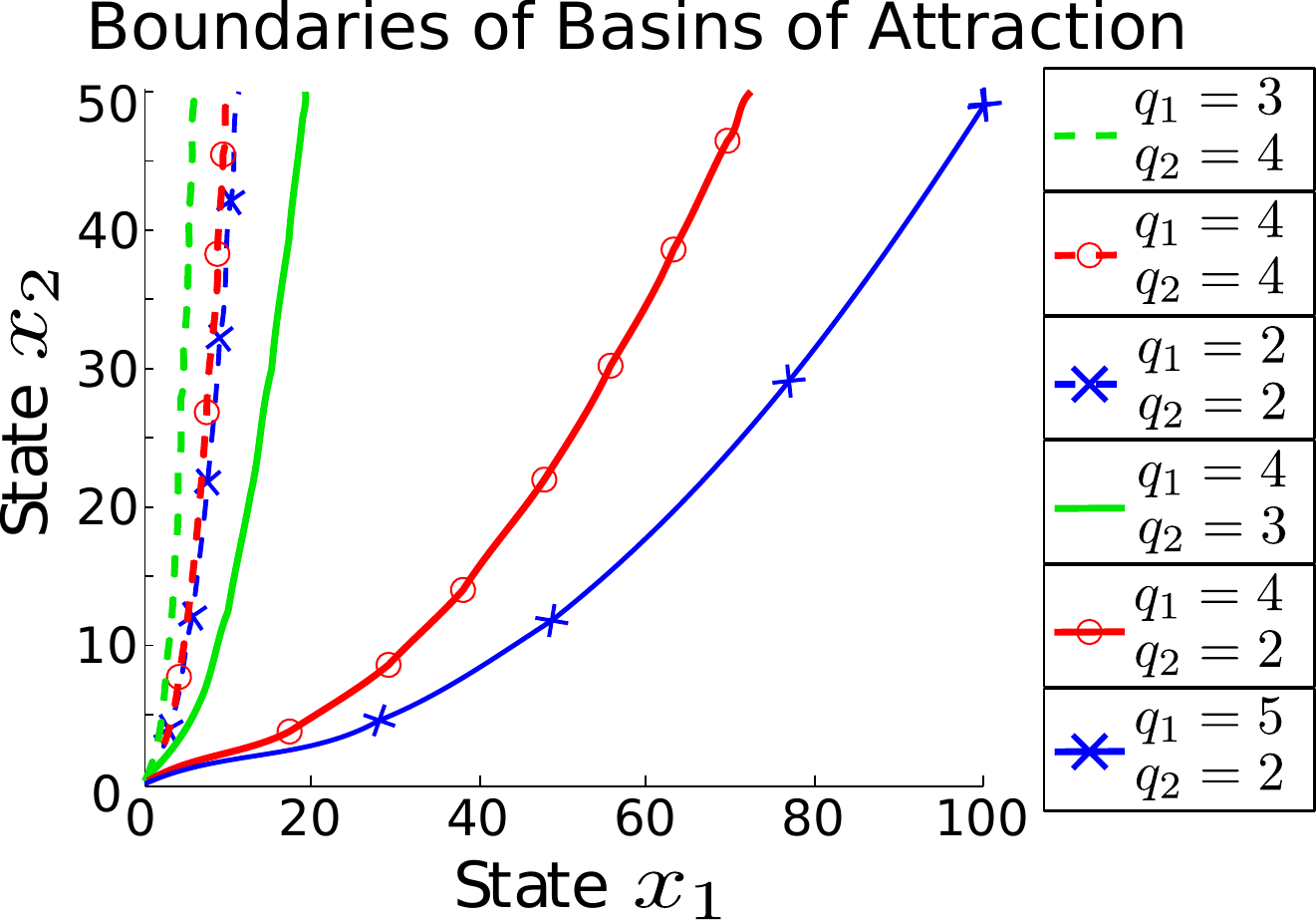}
	\caption{The effect of Hill coefficients (here denoted as $q_1$ and $q_2$) on the domain of attractions of the toggle switch system~\eqref{sys:toggle-2d}. The remaining parameters are fixed and described in~\eqref{param-values-coop}. } \label{fig:toggle-2d-basins-coop} 
\end{figure}
\begin{figure}[b]\centering
	\includegraphics[width = .43\columnwidth]{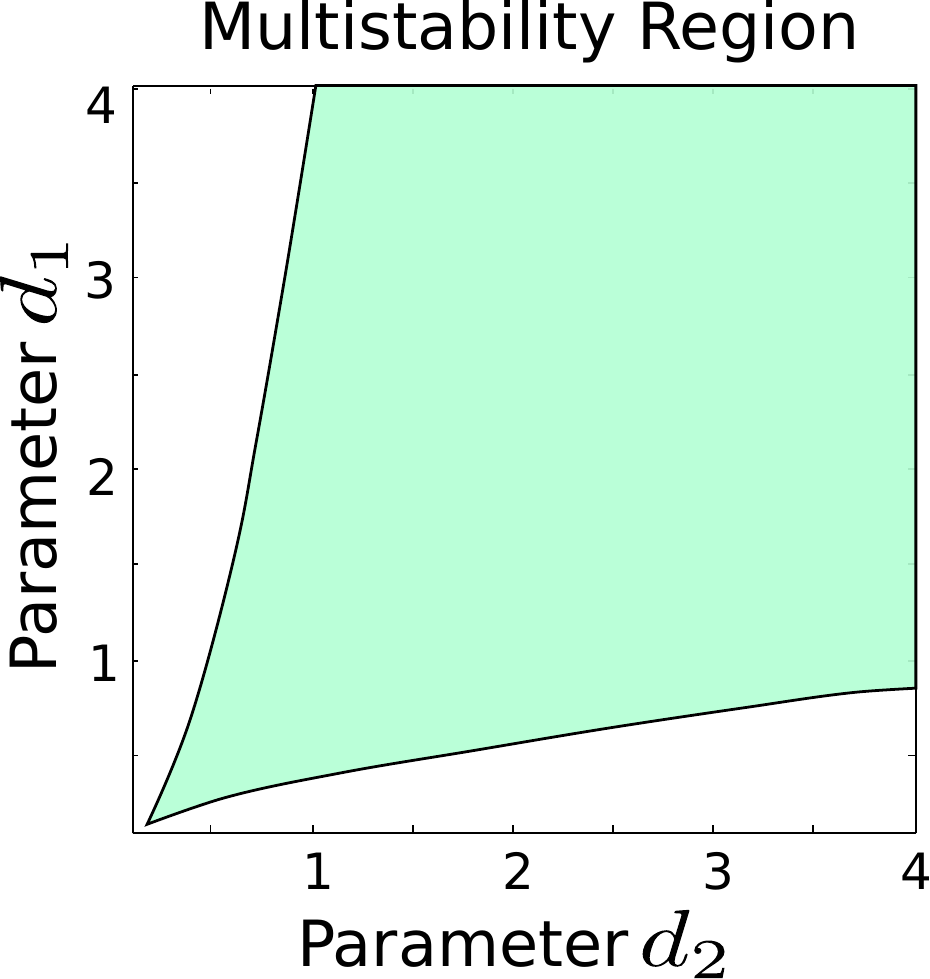}
	\caption{Approximation of the multistability region in the toggle switch system~\eqref{sys:toggle-2d} subject to variations in degradation rates $d_1$, $d_2$ and parameter values~\eqref{param-values-bist}.  The region of monostability is white, while the region of multistability is light green.} \label{fig:param-values-bist} 
\end{figure}

Finally, we illustrate the result of Proposition~\ref{prop:bistability-region} by considering another important problem for a successful design of a toggle switch: estimating the set of parameters for which the system is bistable. Consider the following parameter values
\begin{gather}\label{param-values-bist}
 p = \begin{pmatrix}
 2 & 700 & 2 & d_1 \\ 1 & 1000 & 2 & d_2
 \end{pmatrix}.
\end{gather}
The partial order in the parameter space is induced by the orthant $\diag(1,~-1)\R^2$. We have computed the multistability region on a mesh grid in $[0, 4]^2$, while computing the number of fixed points for given parameter values, and plot the estimation results in Figure~\ref{fig:param-values-bist}. This picture confirms some of the previous findings (including the ones in~\cite{Gardner00}). 

We note that most of the properties described above are known. We show, however, that these properties stem from monotone systems theory and are not limited to a particular model. Hence they can be extended to other models as long as they satisfy the premise of our theoretical results.

\subsection{Basin of Attraction of a Non-Monotone System} \label{ex:laci-tetr-3d}
We consider the following three-state system
\begin{equation}
\label{eq:non-mon}
\begin{array}{rl}
\dot x_1 &= \dfrac{1000}{1 + x_3^{2}} - 0.4 x_1, \\
\dot x_2 &= \dfrac{1000}{1 + x_1^{4}}  - 4 x_2 + u, \\
\dot x_3 &= p_{1}  + p_{2} \frac{x_1}{x_1+1}+ 5 x_2  - 0.3 x_3
\end{array}
\end{equation}
which is not monotone with respect to any order for all positive parameter values $p_1$ and $p_2$. This model does not have any biological interpretation and was designed in order to illustrate Theorem~\ref{thm:comp-sys-b} on a simple example. We consider the nominal system $\cF$ (with $p_1 =0.1$, $p_2 = 1$) and two bounding systems $\cG_1$ (with $p_1 = 0.1$, $p_2 =0$) and  $\cG_2$ ($p_1 = 1.1$, $p_2 =0$).
\begin{figure}[t]
	\centering
	\includegraphics[width = 0.65\columnwidth]{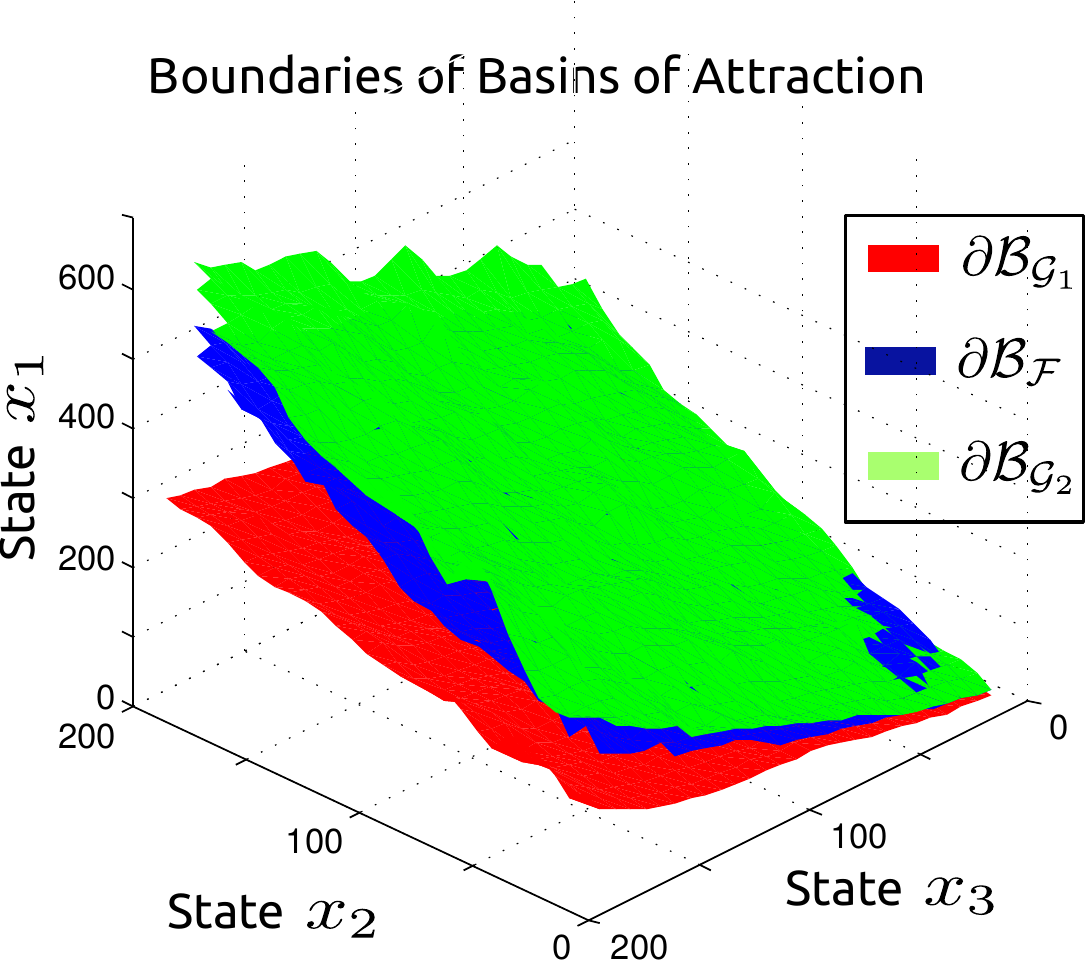}
	\caption{Illustration to Example~\ref{ex:laci-tetr-3d}. Boundaries of domains of attraction $\cB$ for the non-monotone system $\cF$ and its bounding monotone systems $\cG_1$, $\cG_2$. }
	\label{fig:ts-separ}
\end{figure}

It can be verified that $\cG_1$ and  $\cG_2$ are indeed monotone, bistable and bound $\cF$, since $0 \le \frac{x_1}{x_1+1} \le 1$. In this example, we illustrate the application of Theorem~\ref{thm:comp-sys-b}, according to which we can bound the basin of attraction of the non-monotone system $\cF$ by the basin of attraction of the monotone systems $\cG_1$, $\cG_2$. We also illustrate the effectiveness of our computational algorithm. We intentionally do not let the algorithm converge to an accurate solution. We generate only $1000$ samples in order to compute every surface. In this case, $550-650$ samples (depending on an example) are used to compute upper and lower bounds on each surface (we depict only upper bounding surfaces, which are built using $250$-$350$ samples). We see in Figure~\ref{fig:ts-separ} that the green surface related to $\cB_{\cG_2}$ does not bound the blue surface related to $\cB_{\cF}$ for all points in the state-space, which happens since our algorithm did not have enough samples to converge. However, the algorithm is already capable of producing surfaces that depict the overall shape of boundaries of domains of attraction. With more samples, we obtain a more accurate picture, which validates our theoretical results. The computational time for this example is between $10$ to $20$ seconds depending on the computational tolerance of the differential equation solver. 

\subsection{Toxin-Antitoxin System} \label{s:examples-ta}
Consider the toxin-antitoxin system studied in~\cite{cataudella2013conditional}:
\begin{align*}
\dot T &= \frac{\sigma_T}{\left(1 + \frac{[A_f][T_f]}{K_0}\right)(1+\beta_M [T_f])} - \dfrac{1}{(1+\beta_C [T_f])} T \\
\dot A &= \frac{\sigma_A}{\left(1 + \frac{[A_f][T_f]}{K_0}\right)(1+\beta_M [T_f])} - \Gamma_A  A \\
\epsilon [\dot A_f] &= A - \left([A_f] + \dfrac{[A_f] [T_f]}{K_T} + \dfrac{[A_f] [T_f]^2}{K_T K_{T T}}\right) \\ 
\epsilon [\dot T_f] &= T - \left([T_f] + \dfrac{[A_f] [T_f]}{K_T} + 2 \dfrac{[A_f] [T_f]^2}{K_T K_{T T}}\right), 
\end{align*}
where $A$ and $T$ is the total number of toxin and antitoxin proteins, respectively, while $[A_f]$, $[T_f]$ is the number of free toxin and antitoxin proteins. In~\cite{cataudella2013conditional} the model was considered with $\epsilon = 0$. Here, we set $\epsilon = 10^{-9}$ in order to show that, in the case of systems admitting a time-scale separation, we can estimate the basins of attraction of the reduced model without performing model reduction. For the parameters
$\sigma_T = 166.28$, $K_0 = 1$, $\beta_M = \beta_c =0.16$, $\sigma_A = 10^2$, 
$\Gamma_A = 0.2$, $K_T = K_{TT} = 0.3$, the system has two stable hyperbolic fixed points:
\begin{gather*}
x^{\bullet} = \begin{pmatrix}  27.1517 &  80.5151 & 58.4429 & 0.0877  \end{pmatrix} \\ 
x^\ast  = \begin{pmatrix} 162.8103 & 26.2221  &  0.0002 & 110.4375 \end{pmatrix}. 
\end{gather*}

Although the full and reduced systems are not monotone with respect to any orthant, numerical results in~\cite{sootla2015koopman} indicate that the basins of attraction of a reduced order system are still order-convex with respect to $\diag\{[1,~-1]\}\Rnn^2$.
\begin{figure}[t]\centering
	\includegraphics[width = 0.6\columnwidth]{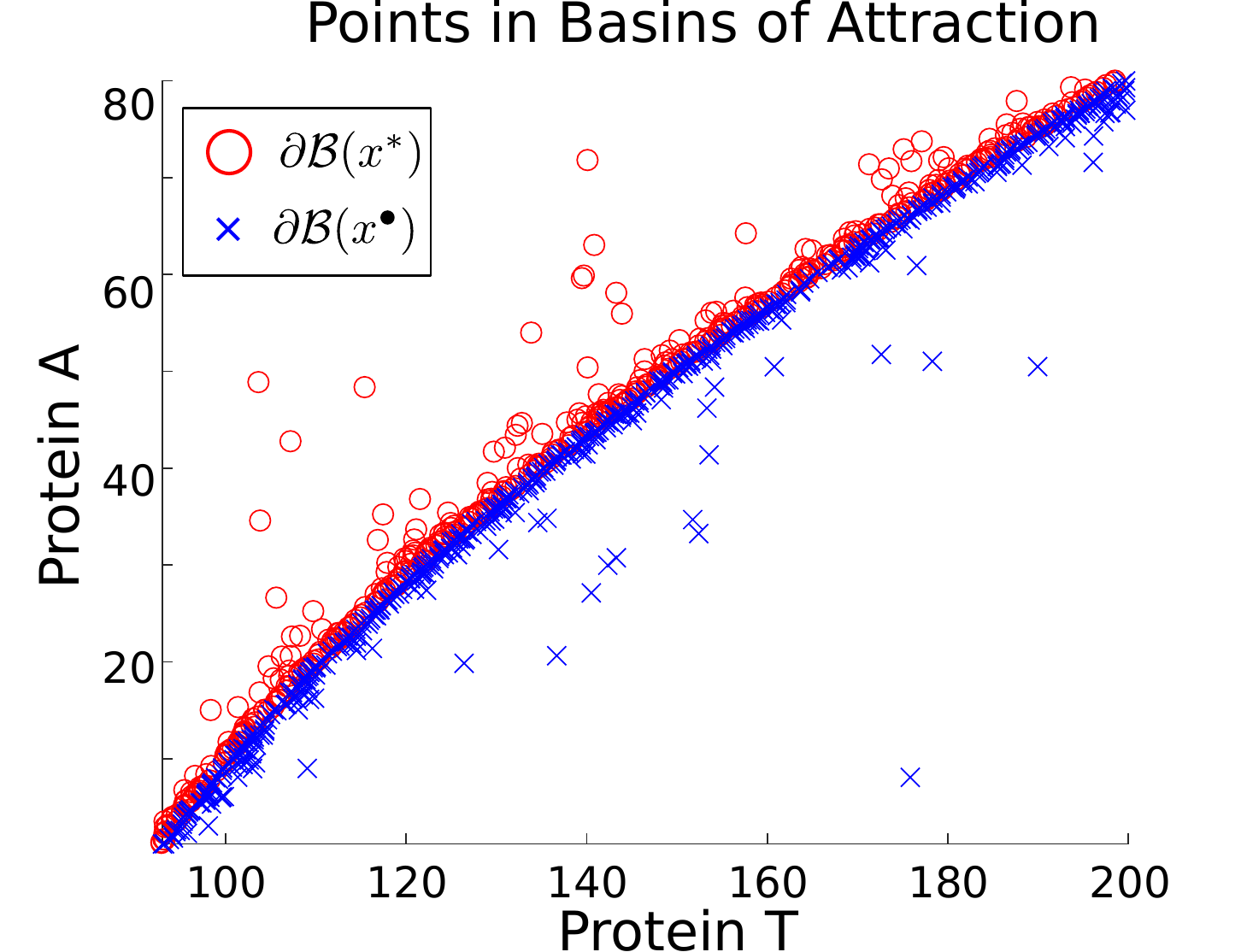}
	\caption{The red circles belong to the cross-section $\widetilde \cB(x^\bullet, \cI, p^1)$  of $\cB(x^\bullet)$, while the blue crosses belong to the cross-section $\widetilde \cB(x^\ast, \cI, p^1)$ of $\cB(x^\ast)$. 
	}\label{fig:tat-bas}
\end{figure}

In this example we compute 2-D cross-sections $\widetilde \cB(x^\ast, \cI, p^i)$ of a basin of attraction of the full order model with $\cI = \{3, 4\}$ and $p^1 = [58.4429,~0.0877]$, $p^2 = [25,~50]$, $p^3=[0.0002,~110.4375]$. All cross-sections gave results indistinguishable by the naked eye, which validates our approach for computing basins of attraction of slow dynamics only. In Figure~\ref{fig:tat-bas}, we present our numerical results for the computation of cross-sections $\widetilde \cB(x^\ast, \cI, p^1)$, $\widetilde \cB(x^\bullet, \cI, p^1)$. We plot all the points generated by our algorithm, including the ones that were pruned out during the algorithm. The total number of generated points was $1000$. The red circles belong to $\widetilde \cB(x^\bullet, \cI, p^1)$, while the blue crosses belong to $\widetilde \cB(x^\ast, \cI, p^1)$. Note that the vast majority of samples is generated near the separatrix between the basins of attraction. We also observe that the red circles and the blue crosses do not violate our assumption on order-convexity of $\cB(x^\ast)$ and $\cB(x^\bullet)$.

\section{Conclusion}
\label{s:conc}
In this paper, we study geometric properties of monotone systems. In particular, we investigate the properties of basins of attraction and relate them to the properties of isostables defined in the framework of the Koopman operator. We discuss in detail the relation between these concepts and their properties under some general assumptions and then focus on properties of basins of attraction of bistable systems. First, we show that we can estimate basins of attraction of bistable non-monotone systems, whose vector fields can be bounded from below and above by bistable monotone systems. This result uses standard tools in monotone systems theory and leads to estimation of basins of attraction of bistable monotone systems under parametric uncertainty. We also discuss a complementary problem: finding the set of parameter values for which a monotone system is (at least) bistable.

We discuss a numerical method for computing inner and outer approximations of basins of attraction of monotone systems. This method exploits the geometric properties of monotone systems and uses the trajectories of the system for computation. The method is potentially well-suited to high-dimensional spaces since it can be easily parallelized and has lower memory requirements than optimization-based methods. We also show how our theoretical results can be used to design a bistable toggle switch with two states. We discuss the effect of different parameters on the shape of the basin of attraction and provide some simple strategies to predict the possible shape of a basin without explicitly computing the basin itself.
\bibliography{ba_final_online.bbl}
\end{document}